\documentclass{tdp}

\usepackage{cite}
\usepackage{amsmath,amssymb,amsfonts}
\usepackage{algorithmic}
\usepackage{graphicx}
\usepackage{textcomp}

\usepackage{url}

\usepackage{amsthm}
\usepackage{thmtools}
\usepackage{thm-restate}

\usepackage{bm}

\usepackage{color}
\usepackage{ifthen}

\newif\ifcommentson\commentsontrue

\newif\ifconferenceon\conferenceontrue
\ifconferenceon
\newcommand{\conference}[1]{#1}
\newcommand{\arxiv}[1]{}
\else
\newcommand{\conference}[1]{}
\newcommand{\arxiv}[1]{#1}
\usepackage{balance}
\fi

\allowdisplaybreaks[1]

\newcommand{\nats}{\mathbb{N}}
\newcommand{\reals}{\mathbb{R}}

\newcommand{\nngreals}{\mathbb{R}_{\ge0}}

\renewcommand{\epsilon}{\varepsilon}
\newcommand{\calC}{\mathcal{C}}
\newcommand{\calF}{\mathcal{F}}
\newcommand{\calG}{\mathcal{G}}
\newcommand{\calH}{\mathcal{H}}
\newcommand{\calR}{\mathcal{R}}
\newcommand{\calU}{\mathcal{U}}
\newcommand{\calX}{\mathcal{X}}
\newcommand{\calY}{\mathcal{Y}}

\newcommand{\bmQ}{\mathbf{Q}}
\newcommand{\bmR}{\mathbf{R}}
\newcommand{\bmS}{\mathbf{S}}
\newcommand{\bmX}{\mathbf{X}}
\newcommand{\bmY}{\mathbf{Y}}
\newcommand{\bmc}{\mathbf{c}}

\newcommand{\bmp}{\mathbf{p}}
\newcommand{\bms}{\mathbf{s}}

\newcommand{\bmy}{\mathbf{y}}
\newcommand{\bmpi}{\mathbf{\pi}}
\newcommand{\bmLamb}{\Lambda}

\newcommand{\bmQxy}{\mathbf{Q}(y|x)}
\newcommand{\bmQxdy}{\mathbf{Q}(y|x')}
\newcommand{\bmQRR}{\bmQ_{\it RR}}
\newcommand{\bmQGLH}{\bmQ_{\it GLH}}

\newcommand{\bbE}{\mathbb{E}}

\newcommand{\hbmp}{\hat{\bmp}}

\begin{document}

\title{Toward Evaluating Re-identification Risks in the Local Privacy Model}
\author{Takao Murakami$^{1,*}$, Kenta Takahashi$^{2}$}
\address{$^{1}$AIST, Japan; 
  $^{2}$Hitachi, Japan;
  \\
  $^*$E-mail: {\small \tt{takao-murakami@aist.go.jp}}
}

\TDPRunningAuthors{Takao Murakami et al.}
\TDPRunningTitle{Toward Evaluating Re-identification Risks in the Local Privacy Model}
\TDPThisVolume{14}
\TDPThisYear{2021}
\TDPFirstPageNumber{79}
%\TDPSubmissionDates{Received 29 February 2008; received in revised form 30 February 2008; accepted 1 March 2008}

\maketitle

\begin{abstract}
LDP (Local Differential Privacy) has 
recently attracted much attention as a metric of data privacy that 
prevents 
the inference of personal data from obfuscated data in the local model. 
However, there are scenarios in which 
the adversary 
wants 
to perform \textit{re-identification attacks} to link the obfuscated data to users in this model. 
LDP can cause excessive obfuscation and destroy the utility in these scenarios because it is not designed to directly prevent re-identification. 
In this paper, we propose a 
measure of re-identification risks, 
which we call 
\textit{PIE (Personal Information Entropy)}. 
The PIE is designed so that it directly prevents re-identification attacks in the local model. 
It lower-bounds the lowest possible re-identification error probability (i.e., Bayes error probability) of the adversary. 
We analyze the relation between LDP and 
the PIE, 
and analyze the PIE and utility in distribution estimation for two obfuscation mechanisms providing LDP. 
Through experiments, we show that 
when we consider re-identification as a privacy risk, 
LDP 
can cause excessive obfuscation and destroy the utility. 
Then we show that 
the PIE can be used to guarantee low re-identification risks for the local obfuscation mechanisms 
while keeping high utility. 
\end{abstract}

\begin{keywords}
Bayes error, distribution estimation, local privacy model, re-identification, user privacy
\end{keywords}

\section{Introduction}
\label{sec:intro}
With the widespread use of personal computers, mobile devices, and IoT (Internet-of-Things) devices, a great amount of personal data 
(e.g., location data \cite{Zheng_WWW09}, rating history data \cite{RecommenderSystems_book}, browser settings \cite{Erlingsson_CCS14}) 
are increasingly collected and used for 
data analysis. 
However, 
the collection of personal data can raise serious privacy concerns. 
For example, users' sensitive locations (e.g., hospitals, stores) 
can be estimated from their location traces (time-series location trails). 
Even if location traces are pseudonymized, they 
can be re-identified (de-anonymized) to link location traces with user IDs \cite{Gambs_JCSS14,Mulder_WPES08,Murakami_TIFS17}. 
An anonymized rating dataset can also be de-anonymized to learn sensitive ratings of users \cite{Narayanan_SP08}. 

DP (Differential Privacy) \cite{Dwork_ICALP06,DP} is a privacy metric that protects users' privacy against adversaries with arbitrary background knowledge, and is known as a gold standard for data privacy. 
According to the underlying architecture, DP can be divided into 
the centralized DP and LDP (Local DP). 
The centralized DP assumes 
a 
centralized model, in which a trusted data collector, who has access to all user's personal data, obfuscates the data. 
On the other hand, LDP assumes 
a 
local model, in which a user obfuscates her personal data by herself and sends the obfuscated data to a (possibly malicious or untrustworthy) data collector. 
While all user's personal data can be leaked from the data collector by illegal access in the centralized model, LDP does not suffer from such data leakage. 
Thus 
LDP has been recently studied in the literature 
\cite{Bassily_STOC15,Cormode_SIGMOD18,Fanti_PoPETs16,Kairouz_ICML16,Murakami_USENIX19,Qin_CCS16,Wang_USENIX17} and has been adopted by 
several industrial applications 
\cite{Erlingsson_CCS14,Ding_NIPS17,Thakurta_USPatent17}. 

However, LDP is designed as a metric of \textit{data privacy} that aims to prevent the inference of personal data 
(i.e., guarantee the indistinguishability of the original data), and there are scenarios in which we should consider \textit{user privacy} 
that aims to prevent re-identification 
in the local model. 
Below we present two examples of the scenarios. 

The first example is an \textit{application that does not require user IDs}. 
The main application of LDP is estimating \textit{aggregate statistics} such as a distribution of personal data \cite{Fanti_PoPETs16,Kairouz_ICML16,Murakami_USENIX19,Wang_USENIX17} and heavy hitters \cite{Bassily_STOC15,Wang_USENIX17}. 
In this case, what are needed are each user's obfuscated data (e.g., noisy locations, noisy purchase history), and her user ID does not have to be collected. 
In fact, some applications (e.g., Google Maps, Foursquare, YouTube recommendations) can be used without 
requiring a user login. 
In such applications, the adversary (who can be either the data collector or outsider) obtains only the obfuscated data, and 
wants 
to perform a \textit{re-identification attack} to identify the user who has sent the obfuscated data.

The second example 
is \textit{pseudonymization}. 
Suppose a mobile application which sends both the user ID and personal (or obfuscated) data to the data collector. 
The data collector can \textit{pseudonymize} all personal (or obfuscated) data to reduce the risks to the users, as described in GDPR \cite{GDPR}. 
In this case, an outsider adversary who obtains the personal (or obfuscated) data via illegal access has to re-identify the user. 
Despite the importance of re-identification risks in the local model, a metric of re-identification risks in this model has not been well established (see Section~\ref{sec:related} for details).

One might think that LDP with a small privacy budget $\epsilon$ (e.g., $\epsilon \leq 1$ \cite{DP_Li}) is enough to prevent re-identification attacks because it guarantees the indistinguishability of the original personal data. 
In other words, if personal data are obfuscated so that the adversary cannot infer the original data, then it seems to be impossible for the adversary to re-identify the user. 
This is indeed the case -- we also show that LDP with a small privacy budget $\epsilon$ prevents re-identificatiion in our experiments. 
However, the real issue of LDP is that it is not designed to directly prevent re-identification, and it makes the original data indistinguishable from any other possible data in the data domain. 
This 
can cause 
excessive obfuscation and 
destroy 
the utility, 
as shown in this paper. 

\smallskip
\noindent{\textbf{Our Contributions.}}~~In this paper, we make the following contributions:

\smallskip
\noindent{\textbf{1) PIE (Personal Information Entropy).}}~~We 
propose 
a 
measure 
of re-identification risks in the local model, which we call the \textit{PIE (Personal Information Entropy)}. 
The PIE is given by the mutual information between a user and (possibly obfuscated) personal data. 
The PIE is applicable to any kind of personal data, and 
does not specify an identification algorithm used by an adversary. 
The PIE 
lower-bounds the lowest possible re-identification error probability (i.e., Bayes error probability) of the adversary. 
We also propose a privacy 
metric 
called \textit{PIE privacy} that upper-bounds the PIE irrespective of the adversary's background knowledge. 
We also show that \textit{pseudonymization (random permutation) alone} guarantees a high re-identification error probability in some cases using our PIE, 
whereas random permutation alone cannot guarantee 
DP 
even using its average versions (e.g., Kullback-Leibler DP \cite{Barber_arXiv14,Cuff_CCS16}, mutual information DP \cite{Cuff_CCS16}) or recently proposed shuffling techniques 
\cite{Balle_arXiv19,Erlingsson_SODA19}.

\smallskip
\noindent{\textbf{2) Theoretical Analysis of the PIE for Obfuscation Mechanisms.}}~~We 
analyze the PIE for two existing 
local obfuscation 
mechanisms: the RR (Randomized Response) for multiple alphabets \cite{Kairouz_ICML16} and the generalized version of local hashing in \cite{Wang_USENIX17}, which we call the GLH (General Local Hashing). 
Both of them are mechanisms providing LDP, and can be used to examine the relationship between LDP and the PIE. 

We 
first 
show that our PIE privacy is a relaxation of LDP; 
i.e., any LDP mechanism provides PIE privacy, hence upper-bounds the PIE. Then we show that this general upper-bound on the PIE for any LDP mechanism is loose, and 
show much tighter upper-bounds on the PIE for the RR and GLH. 

\smallskip
\noindent{\textbf{3) Theoretical Analysis of the Utility for Obfuscation Mechanisms.}}~~We 
analyze the utility of the RR and GLH 
for given PIE guarantees. 
Here, we 
consider 
discrete distribution estimation \cite{Agrawal_PODS01,Agrawal_SIGMOD05,Fanti_PoPETs16,Kairouz_ICML16,Kairouz_JMLR16,Murakami_USENIX19}, where personal data take discrete values, as a task for the data collector.

In our utility analysis, we show that our PIE privacy has a very different implication for utility and privacy than LDP. 
Specifically, 
let $\calX$ be a finite set of personal data. 
Then the GLH reduces the size $|\calX|$ of personal data to $g$ (i.e., dimension reduction) via random projection. When we use LDP as a privacy 
metric, 
the optimal value of $g$ is given by: $g=e^\epsilon+1$ \cite{Wang_USENIX17}, where $\epsilon$ is a 
privacy budget 
of LDP. 
In contrast, we show that when we use PIE privacy as a privacy 
metric, 
a larger $g$ provides better utility. 
In other words, we show an intuitive result that  \textit{compressing the personal data with a smaller $g$ results in the loss of utility} in our privacy 
metric. 
This result is caused by the fact that PIE privacy 
prevents 
the identification of users, whereas LDP 
prevents 
the inference of personal data. 

\smallskip
\noindent{\textbf{4) Evaluating the PIE for Obfuscation Mechanisms.}}~~We evaluate the privacy and utility 
the RR and GLH. 
We first show that LDP destroys utility when the privacy budget $\epsilon$ is small. 
For example, for distribution estimation of the most popular $20$ POIs (Point-of-Interests) in the Foursquare dataset \cite{Yang_WWW19}, the relative error of the RR was $1.05$ $(> 1)$ even when $\epsilon=10$ (which is considered to be still fairly large \cite{DP_Li}). 
In other words, LDP fails to guarantee meaningful privacy and utility for this task. 
This comes from the fact that LDP is not designed to directly prevent re-identification attacks. 

We next show that 
the PIE 
can be used to guarantee a low re-identification error while keeping high utility. 
For example, when we used 
the PIE 
to guarantee the re-identification error probability larger than 
$0.92$, 
the relative error of the RR for the top-$20$ POIs was 
$0.10$ 
$(\ll 1)$. 
This suggests that 
when we consider re-identification as a privacy risk, we should design 
a privacy metric that directly prevents re-identification attacks.

\smallskip
\noindent{\textbf{5) PSE (Personal Identification System Entropy).}}~~As explained above, we show upper-bounds on the PIE for the RR and GLH. 
Then a natural question would be ``how tight are these upper-bounds?'' 
To answer to this question, 
we 
introduce the \textit{PSE (Personal Identification System Entropy)}, 
which 
is a \textit{lower-bound} on the PIE and is equal to the PIE under some conditions. 
The PSE is designed to be easily calculated by specifying an 
identification algorithm. 
We show through experiments that our upper-bounds on the PIE for the RR and GLH are close to the PSE, which indicates that our upper-bounds are fairly tight and cannot be improved much.

One additional 
interesting feature of 
the PSE 
is that it 
can be used to compare the identifiability of personal data such as location traces and rating history with the identifiability of \textit{biometric data} such as a fingerprint and face.
Nowadays biometric authentication is widely used for various applications such as unlocking a smartphone, 
banking, and physical access control. 
Our 
PSE 
provides 
a 
new intuitive understanding 
of re-identification risks 
by comparing two different sources of information. 

Note that we 
do not consider privacy risks or obfuscation (e.g., adding DP noise) for biometric data. 
Our interest here is intuitive understanding of the identifiability of personal data (e.g., locations, rating history) through the comparison with biometric data.

For example, we show that 
a location trace with at least $500$ locations has higher identifiability than the face matcher in \cite{BSSR}. 
We also note that the face dataset in \cite{BSSR} has lower errors than the best matcher in the FRPC (Face Recognition Prize Challenge) 2017 \cite{Grother_NISTIR17} where face images were collected without tight quality constraints. 
In other words, we reveal the fact that \textit{the location trace is more identifiable than the face matcher that won the 1st place in the FRPC 2017}. 
We believe that this result is of independent interest. 

    %higher identifiability than a face (see below for more details).}
    
    %Based on this framework, we propose two unified measures of identifiability: the PIE and PSE. 
    % The PIE and PSE are generalizations of the BI (Biometric Information) for a system \cite{Adler_PAP09} and BSE (Biometric System Entropy) \cite{Takahashi_IMAVIS14}, respectively. 
    % The PIE is an intuitive measure given by the mutual information between a user and 
    % (possibly obfuscated) personal data. 
    % It also does not assume a specific adversary. 
    % In contrast, the PSE is designed to be easily evaluated through experiments by specifying an adversary. 
    % \item We propose a privacy notion called \textit{PIE privacy} that upper-bounds the PIE, 
    % and 
    % show that it 
    % is a relaxation of LDP; i.e., any LDP mechanism provides PIE privacy.  
    % The key difference between these two notions is that 
    % PIE privacy is \textit{user privacy} which aims to prevent the identification of users, whereas LDP is \textit{data privacy} which aims to prevent the inference of personal data. 
    % We also show that PIE privacy is more suitable to prevent the identification of users.
    % \item 

\smallskip
\noindent{\textbf{Remark.}}~~Our PIE privacy is based on the mutual information, which quantifies the \textit{average} amount of information about a user through observing (possibly obfuscated) personal data. 
Thus our PIE privacy is an average privacy notion, as with the KL (Kullback-Leibler)-DP \cite{Barber_arXiv14,Cuff_CCS16} and the mutual information DP \cite{Cuff_CCS16}. 
Average privacy notions have also been used in some studies on location privacy \cite{Romanelli_arXiv20,Shokri_SP11,Shokri_CCS12}. 

The caveat of the average privacy notion is that it may not guarantee the 
indistinguishability 
for every user; 
e.g., even if it guarantees the re-identification error probability larger than $0.99$, at most $1\%$ of users may be re-identified. 
Nevertheless, 
there are application scenarios in which 
the average privacy notion is quite useful in practice. 
One example is a 
prioritization system that determines, among several defenses, which one should be (or should not be) used. 
Another example is an alerting system adopted in \cite{Chia_SP19}, which notifies engineers if re-identification risks exceed pre-determined limits.
Thus, we use the average privacy notion as a starting point.

We also note that 
a \textit{worst-case} privacy notion (e.g., 
min-entropy \cite{Smith_FoSSaCS09}) 
can be used to guarantee stronger privacy for every user. 
We leave extending our PIE privacy to the worst-case notion for future work (Section~\ref{sec:conclusion} describes some open questions in this research direction).

\smallskip
\noindent{\textbf{Basic Notations.}}~~Let $\nats$, $\reals$, and $\nngreals$ be the set of natural numbers, real numbers, and non-negative real numbers, respectively. 
For $a \in \nats$, let $[a] = \{1, 2, \cdots, a\}$. 
For random variables $A$ and $B$, let $I(A; B)$ be the mutual information between $A$ and $B$. 
For two distributions $p$ and $q$, let $D(p||q)$ be the 
KL (Kullback-Leibler) divergence \cite{elements}. 
We simply denote the logarithm with base 2 by $\log$.
We use these notations throughout this paper.

\section{Related Work}
\label{sec:related}
In this section, we review the previous work. 
Section~\ref{sub:privacy_measures} describes 
LDP \cite{Duchi_FOCS13} 
and other privacy metrics. 
Section~\ref{sub:obf} explains the RR for multiple alphabets \cite{Kairouz_ICML16} and the GLH \cite{Bassily_STOC15}.

\subsection{Privacy Metrics}
\label{sub:privacy_measures}
\noindent{\textbf{LDP.}}~~Let 
$\calX$ be a finite set of personal data, and $\calY$ be a finite set of obfuscated data. 
Let $\bmQ$ be an obfuscation mechanism (a.k.a. masking method \cite{Torra_book}), which maps personal data $x\in\calX$ to obfuscated data $y\in\calY$ with probability $\bmQ(y|x)$.  
Then 
LDP is defined as follows:

\begin{definition}[$\epsilon$-LDP]
\label{def:LDP}
Let $\epsilon \in \nngreals$. 
An obfuscation mechanism $\bmQ$ provides \emph{$\epsilon$-LDP} if for any $x,x' \in \calX$ and any $y \in \calY$, 
\begin{align}
\bmQxy \leq e^\epsilon \bmQxdy.
\label{eq:LDP}
\end{align}
\end{definition}
Intuitively, LDP guarantees that the adversary who obtains $y$ cannot determine whether it comes from $x$ or $x'$ for any pair of $x$ and $x'$ in $\calX$ with a certain degree of confidence. 
The parameter $\epsilon$ is called the privacy budget. 
When the privacy budget $\epsilon$ is close to $0$, all of the data in $\calX$ are almost equally likely. 
Thus LDP strongly protects $y$ 
when $\epsilon$ is small; e.g., $\epsilon \leq 1$ \cite{DP_Li}.

\smallskip
\noindent{\textbf{Other Privacy Metrics.}}~~To date, 
numerous variants of DP (or LDP) have been proposed to provide different types of privacy guarantees. 
Examples are: Pufferfish privacy \cite{Kifer_TODS14}, $d_x$-privacy \cite{Chatzikokolakis_PETS13}, R\'{e}nyi DP \cite{Mironov_CSF17}, concentrated DP \cite{Dwork_arXiv16}, mutual information DP \cite{Cuff_CCS16}, personalized DP \cite{Jorgensen_ICDE15}, 
utility-optimized LDP \cite{Murakami_USENIX19}, 
and capacity-bounded DP \cite{Chaudhuri_NeuRIPS19}. 
A recent SoK paper proposed a systematic taxonomy of relaxations of DP, and classified the relaxations into seven categories based on which aspect of DP was modified \cite{Desfontaines_PoPETs20}. 

Our 
PIE privacy 
is also a variant of DP because it is a relaxation of LDP, as shown in Sections~\ref{sub:LDP_PIE}. 
Our 
PIE privacy 
differs from existing variants of DP in 
that 
PIE privacy aims at preventing re-identification attacks. 
In this regard, PIE privacy is different from any dimension of seven categories in the SoK paper \cite{Desfontaines_PoPETs20}.

Our 
PIE and PSE 
are also related to quantitative information flow \cite{Alvim_CSF12,Bordenabe_CSF16,Romanelli_arXiv20,Smith_FoSSaCS09}, where the amount of information leakage is measured by the mutual information or entropy. 
In particular, our PSE is closely related to a recent study \cite{Romanelli_arXiv20}, which uses the mutual information between a user ID and a re-identified user ID as a 
measure of re-identification risks. 

Our PSE differs from \cite{Romanelli_arXiv20} in that 
the PSE is given by the mutual information between a user ID and a \textit{score vector} (vector consisting of similarities or distances for all users) 
calculated by the adversary, rather than the re-identified user ID. 
We use the score vector because it contains much richer information than the identified user ID (this is well known in biometrics; e.g., see \cite{Ross06}). 
We also show that 
although the PSE is upper-bounded by the PIE, 
\textit{the PSE is equal to the PIE under some conditions} 
(Section~\ref{sub:PIE}). 
This property is very useful for evaluating how tight an upper-bound on the PIE is. 
In fact, we show that \textit{our upper-bounds on the PIE for the RR and GLH are close to the PSE} in our experiments (Section~\ref{sec:exp}). 
In contrast, it is difficult to evaluate the upper-bounds on the PIE using the re-identification user ID because it contains much less information. We also note that a score vector enables us to output a list of $k \in [n]$ users whose similarities (resp.~distance) are the highest (resp.~lowest) as an identification result.

\subsection{Obfuscation Mechanisms}
\label{sub:obf}
\noindent{\textbf{RR for Multiple Alphabets.}}~~The RR (Randomized Response) was originally introduced 
by Warner 
for binary alphabets \cite{Warner_JASA65}. 
Kairouz \textit{et al.} 
\cite{Kairouz_ICML16} 
studied 
the RR for $|\calX|$-ary alphabets. 

Given $\epsilon \in \nngreals$, let $\bmQ_{RR}$ be the \textit{$\epsilon$-RR} for $|\calX|$-ary alphabets. 
In the $\epsilon$-RR, the output range is identical to the input domain; i.e., $|\calX|=|\calY|$. 
Given personal data $x\in\calX$, the $\epsilon$-RR outputs obfuscated data $y\in\calX$ with 
probability: 
\begin{align}
\bmQRR(y | x) = 
\begin{cases}
 \frac{e^\epsilon}{|\calX|+e^\epsilon-1} & \text{(if $y = x$)}\\
 \frac{1}{|\calX|+e^\epsilon-1} & \text{(otherwise)}.\\
\end{cases} 
\label{eq:RR}
\end{align}
By (\ref{eq:LDP}) and (\ref{eq:RR}), the $\epsilon$-RR provides $\epsilon$-LDP. 

Note that the $\epsilon$-RR is a kind of PRAM (Post-RAndomization Method) \cite{Torra_book,Templ_TDP08,Mares_CSDA14}, which replace a category $x\in\calX$ with another category $y\in\calX$ according to a given transition matrix ($\bmQ$ can be viewed as a transition matrix), and that there is a connection between the literature of DP and that of PRAM. 
For example, Mar\'{e}s and Shlomo \cite{Mares_CSDA14} consider a PRAM transition matrix whose diagonal values 
do not exceed 
$0.75$ to reduce privacy disclosure risks. 
When 
$\epsilon = 1$ (which is considered to be acceptable in the literature of DP \cite{DP_Li}) in the $\epsilon$-RR, the diagonal values in $\bmQRR$ are smaller than $0.73$ for any $|\calX| \geq 2$ (by (\ref{eq:RR})). 

\smallskip
\noindent{\textbf{GLH.}}~~Bassily and Smith \cite{Bassily_STOC15} proposed an obfuscation mechanism based on random projection that maps personal data $x\in\calX$ to a single bit. 
Wang \textit{et al.} \cite{Wang_USENIX17} called this mechanism the BLH (Binary Local Hashing), and generalized it so that $x\in\calX$ is mapped to a value in $[g]$, where $g\in\nats$.
We 
call 
this generalized mechanism 
the \textit{GLH (General Local Hashing)}. 

The GLH consists of the following two steps: (i) apply random projection to $x\in\calX$, 
and then (ii) perturb the data using the RR for multiple alphabets. 
Formally, 
let $\calH=\{h: \calX \rightarrow [g]\}$ be a universal hash function family \cite{Carter_JCSS79}; i.e., for any distinct $x$ and $x'$ in $\calX$ and a hash function $h$ chosen uniformly at random in $\calH$, 
\begin{align*}
\Pr [h(x) = h(x')] \leq \frac{1}{g}. 
\end{align*}
Given $\epsilon \in \nngreals$, let $\bmQGLH$ be the \textit{$(g,\epsilon)$-GLH}. 
The $(g,\epsilon)$-GLH is an obfuscation mechanism with the input alphabet $\calX$ and the output alphabets $\calY=(\calH,[g])$. 
Given $x\in\calX$, the $(g,\epsilon)$-GLH randomly generates a hash function $h$ from $\calH$, and outputs $(h,y)\in\calY$ with 
probability: 
\begin{align}
\bmQGLH((h,y) | x) = 
\begin{cases}
 \frac{e^\epsilon}{g+e^\epsilon-1} & \text{(if $y = h(x)$)}\\
 \frac{1}{g+e^\epsilon-1} & \text{(otherwise)}.\\
\end{cases} 
\label{eq:GLH}
\end{align}
By (\ref{eq:LDP}) and (\ref{eq:GLH}), the $(g,\epsilon)$-GLH provides $\epsilon$-LDP. 
Wang \textit{et al.} \cite{Wang_USENIX17} found that for a fixed $\epsilon$, the value of $g$ that minimizes the variance in distribution estimation is given by: $g=e^\epsilon+1$. 

\section{Personal Information Entropy}
\label{sec:unified}
We propose the PIE (Personal Information Entropy) as a 
measure 
of re-identification risks in the local model. 
We first describe 
frameworks for obfuscation and identification 
assumed 
in our work in Section~\ref{sub:framework}. 
Then in Section~\ref{sub:PIE}, we introduce the PIE and a privacy metric called PIE privacy, which upper-bounds the PIE. 
We also introduce 
the PSE (Personal Identification System Entropy), which 
lower-bounds the PIE by specifying an 
identification algorithm. 
Finally we show several basic properties of the PIE 
in Section~\ref{sub:properties}.

\subsection{Obfuscation/Identification Framework}
\label{sub:framework}

\noindent{\textbf{Framework.}}~~Figures~\ref{fig:framework_obf} and \ref{fig:framework_ide} show 
an obfuscation framework and identification framework, 
respectively. 
We also show in Table~\ref{tab:notations} the notations used in this paper. 
Let $\Omega$ be a finite set of all human beings, and 
$\calU \subseteq \Omega$ be a finite set of users who use a certain application; e.g., location-based service, recommendation service. 
Let $n \in \nats$ be the number of users in $\calU$, and $u_i \in \calU$ be the $i$-th user; i.e., $\calU = \{u_1, \cdots, u_n\}$. 

\begin{figure}
\centering
\includegraphics[width=0.75\linewidth]{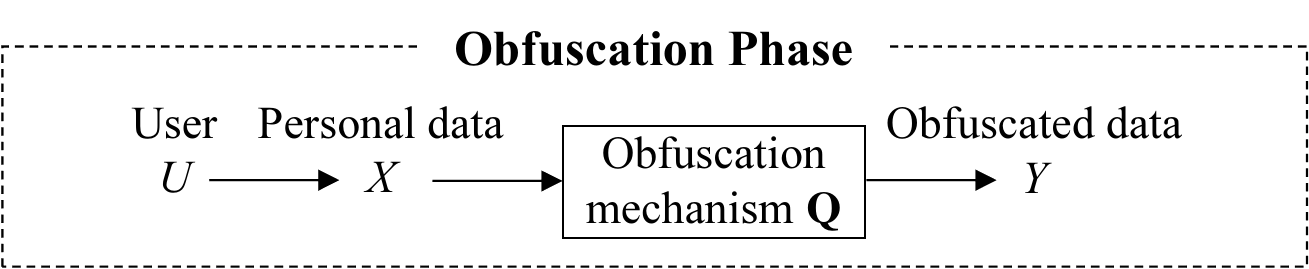}
\vspace{-2mm}
\caption{Framework for obfuscation. 
Our PIE is built upon this obfuscation framework. 
The PIE does not specify any identification algorithm.}
\label{fig:framework_obf}
\includegraphics[width=0.75\linewidth]{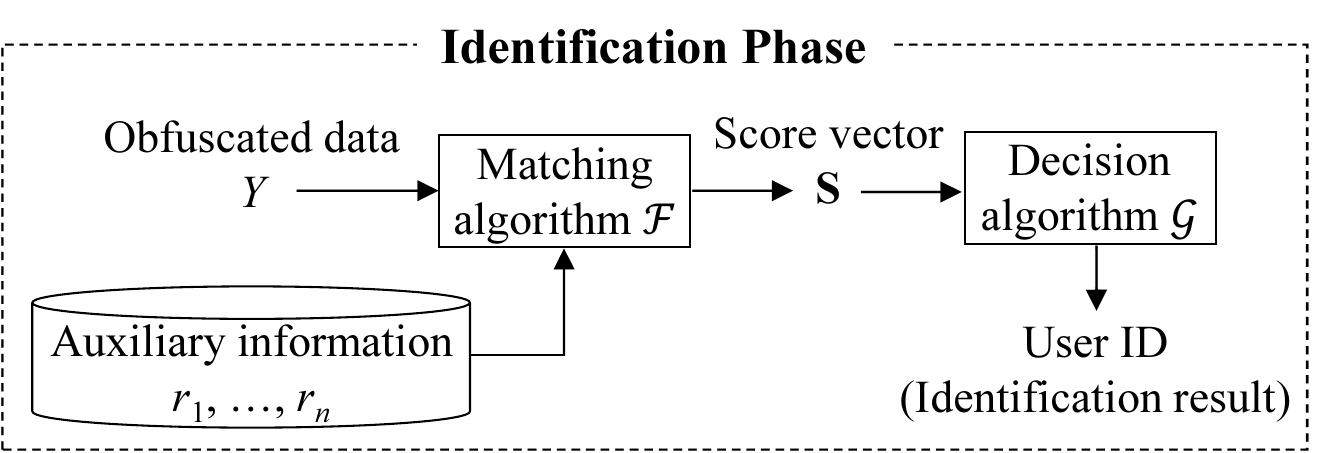}
\vspace{-2mm}
\caption{Framework for identification. 
A score vector $\bms=(s_1, \cdots, \allowbreak s_n)\in\reals^n$ consists of a numerical score (similarity or distance) $s_i\in\reals$ between obfuscated data $Y$ and auxiliary information $r_i$ of user $u_i$.
Our PSE is built upon this identification framework.}
\label{fig:framework_ide}
\end{figure}

\begin{table}[tp]
\caption{Notations in this paper.} 
\centering
\hbox to\hsize{\hfil
\begin{tabular}{l|l}
\hline
Symbol		&	Description\\
\hline
$\calU$	    &	Finite set of users.\\
$\calX$	    &	Finite set of personal data.\\
$\calY$	    &	Finite set of obfuscated data.\\
$\calR$	    &	Finite set of auxiliary data.\\
$n$ 	    &	Number of users in $\calU$ ($n\in\nats$).\\
$u_i$	    &	$i$-th user ($u_i\in\calU$).\\
$r_i$	    &	Auxiliary information of user $u_i$.\\
$s_i$	    &	Score of user $u_i$ ($s_i\in\reals$).\\
$\bms$	    &	Score vector ($\bms=(s_1, \cdots, \allowbreak s_n)\in\reals^n$).\\
$U$ 	    &	Random variable representing a user in $\calU$.\\
$X$ 	    &	Random variable representing personal data.\\
$Y$ 	    &	Random variable representing obfuscated data.\\
$\bmS$	    &	Random variable representing a score vector.\\
$p_U$       &   Distribution of $U$.\\
$p_{X|U=u_i}$   &   Distribution of $X$ given $U=u_i$.\\
$p_{U,X}$   &   Joint distribution of $U$ and $X$.\\
$\bmQ$	    &	Obfuscation mechanism.\\
$\calF$	    &	Matching algorithm.\\
$\calG$	    &	Decision algorithm.\\
\hline
\end{tabular}
\hfil}
\label{tab:notations}
\end{table}

In the obfuscation phase, 
a user 
obfuscates her personal data 
by herself using an obfuscation mechanism $\bmQ$ in her device, and sends the obfuscated data 
to a data collector. 
Let $U$, $X$, and $Y$ be random variables representing a user in $\calU$, personal data in $\calX$, and obfuscated data in $\calY$, respectively. 

User $U$ is randomly generated from some distribution over $\calU$, which can be either uniform or non-uniform. 
Let $p_U$ be a distribution of $U$; i.e., 
$p_U(u_i) = \Pr(U=u_i)$. 
Note that 
$p_U$ is a prior distribution an adversary has before observing $Y$. 
In other words, $p_U$ is a kind of the adversary's background knowledge. 
For example, 
suppose that each user sends a single obfuscated datum to a data collector, and that Alice's obfuscated data $Y$ is leaked to the adversary. 
In this case, 
the prior distribution $p_U$ for this adversary is uniform over $\calU$; i.e., 
$p_U(u_i) = \frac{1}{|\calU|}$ for any $u_i\in\calU$.
If some heavy users send obfuscated data many times and the adversary knows this fact, 
then $U$ is non-uniformly distributed. 

Given $U=u_i$, personal data $X$ is randomly generated from some distribution over $\calX$. 
Let $p_{X|U=u_i}$ be a distribution of $X$ given $U=u_i$; i.e., 
$p_{X|U=u_i}(x) = \Pr(X=x | U=u_i)$. 
If personal data $X$ is uniquely determined given user $u_i$ (i.e., 4-digit PIN of Alice is ``3928''), then $p_{X|U=u_i}$ is a point distribution that has probability $1$ for a single data point. 
Otherwise (e.g., Alice may visit many locations), $p_{X|U=u_i}$ is not a point distribution. 

Let $p_{U,X}$ be a joint distribution of $U$ and $X$; i.e., $p_{U,X}(u_i,x) = \Pr(U=u_i, X=x)$. 
Note that $p_{U,X}(u_i,x) = p_U(u_i) p_{X|U=u_i}(x)$. 
As with $p_U$, $p_{U,X}$ is also a kind of the adversary's background knowledge. 
Our PIE depends on 
$p_{U,X}$, 
but our PIE privacy does \textit{not} depend on 
$p_{U,X}$, 
as described in Section~\ref{sub:PIE} in detail. 

Given $X = x$, obfuscated data $Y$ is randomly generated using an obfuscation mechanism $\bmQ$, which maps $x$ to $y\in\calY$ with probability $\bmQ(y|x)$. 
Examples of $\bmQ$ include the randomized response \cite{Warner_JASA65,Kairouz_ICML16}, GLH \cite{Wang_USENIX17}, and RAPPOR \cite{Erlingsson_CCS14}. 
If an obfuscation mechanism $\bmQ$ is not used, then $Y=X$. 
Then the data collector estimates some aggregate statistics; e.g., histogram, heavy hitters. 

Our PIE is built upon the obfuscation framework in Figure~\ref{fig:framework_obf}, and is 
given by the mutual information between $U$ and $Y$, 
as described in Section~\ref{sub:PIE}. 
Therefore, the PIE 
does not specify 
any identification algorithm.

On the other hand, our PSE is designed to evaluate a lower-bound on the PIE through experiments by specifying an 
identification algorithm. 
The PSE is built upon the identification framework in Figure~\ref{fig:framework_ide}, where 
an identification system 
comprises 
two algorithms: a \textit{matching algorithm} and \textit{decision algorithm}. 
These algorithms are widely 
used for 
re-identification 
attacks 
in privacy literature  
\cite{Frankowski_SIGIR06,Gambs_JCSS14,Mulder_WPES08,Murakami_TIFS17,Narayanan_SP08,Shokri_SP11,Shokri_PETS11}. 

In the identification phase, 
the matching algorithm takes as input obfuscated data $Y$ and some \textit{auxiliary information} of user $u_i$, and outputs a \textit{numerical score} 
for $u_i$. 
The auxiliary information is background knowledge about user $u_i$ the system possesses. 
The score is either a \textit{similarity} or \textit{distance} between the obfuscated data $Y$ and 
auxiliary information of $u_i$. 
A large similarity (or small distance) indicates that it is highly likely that 
$Y$ 
belongs to $u_i$. 

For the auxiliary information, we can consider two possible models: the maximum-knowledge model and partial-knowledge model \cite{Domingo-Ferrer_PST15,Ruiz_PSD18}. 
The maximum-knowledge model assumes a worst-case scenario 
(though it may not be realistic). 
Specifically, 
this model assumes that 
the original personal data $X$ is used as auxiliary information. 
The partial-knowledge model considers a 
scenario 
where the adversary does not know the original personal data. 
For example, 
the auxiliary information in this model can be 
locations or video browsing history (other than $X$) disclosed by the users via SNS (e.g., Foursquare, Facebook). 
The de-anonymization attack against the Netflix Prize dataset \cite{Narayanan_SP08} also assumes that the adversary knows only a little bit about a user's rating history (e.g., two or three ratings per user) as auxiliary information, and therefore falls into the partial-knowledge model. 

Formally, let $\calR$ be the set of auxiliary information, and $r_i \in \calR$ be auxiliary information of user $u_i$. 
Let $\calF: \calY \times \calR \rightarrow \reals$ be the matching algorithm, which takes as input obfuscated data $y \in \calY$ and auxiliary information $r_i \in \calR$, and outputs a score $\calF(y,r_i) \in \reals$. 
Examples of $\calF$ include the algorithms based on the Markov 
model \cite{Gambs_JCSS14,Mulder_WPES08,Murakami_TIFS17,Shokri_PETS11}, TF-IDF \cite{Frankowski_SIGIR06}, and the cosine similarity measure \cite{Narayanan_SP08}. 
Let $s_i\in\reals$ be a score of user $u_i$; i.e., $s_i = \calF(y,r_i)$. 
Let $\bms=(s_1, \cdots, \allowbreak s_n)\in\reals^n$ be a \textit{score vector}, and 
$\bmS$ be a random variable representing a score vector. 

The decision algorithm decides who the user is based on a score vector. 
Formally, 
let $\calG: \reals^n \rightarrow \calU$ be the decision algorithm, which takes a score vector $\bms\in\reals^n$ as input and outputs an identified user ID $\calG(\bms)\in\calU$. 
Typically, the decision algorithm $\calG$ outputs a user ID whose similarity (resp.~distance) is the highest (resp.~lowest) \cite{intro,guide,Frankowski_SIGIR06,Gambs_JCSS14,Mulder_WPES08,Murakami_TIFS17,Narayanan_SP08,Shokri_PETS11}. 
We refer to this as a \textit{best score rule}. 
The decision algorithm may also output a list of $k \in [n]$ users whose similarities (resp.~distance) are the highest (resp.~lowest).

\smallskip
\noindent{\textbf{Interpretation as Biometric Identification.}}~~Readers might have noticed that 
the frameworks in Figures~\ref{fig:framework_obf} and \ref{fig:framework_ide} include 
biometric identification as a special case. 
Specifically, 
the obfuscation mechanism $\bmQ$ can be interpreted as a feature extractor in the context of biometrics. 
The auxiliary information $r_1, \cdots, r_n$ can be interpreted as biometric templates enrolled in the database. 
The matching algorithm and decision algorithm are commonly used in biometric identification \cite{intro,guide}. 
Since our PSE is 
built on 
these frameworks, 
it 
can also be applied to measure the identifiability of biometric data.

We emphasize again that we 
do not consider privacy risks or obfuscation (e.g., adding DP noise) for biometric data. 
Instead, we provide 
a new perspective about the identifiability of personal data through the comparison with another source of information; e.g., the best face matcher in the prize challenge, as described in Section~\ref{sec:intro}.

\subsection{PIE and PSE}
\label{sub:PIE}

\smallskip
\noindent{\textbf{PIE.}}~~We 
now 
introduce the \textit{PIE (Personal Identification Entropy)} of user $U$ obtained through 
obfuscated data $Y$. 
Specifically, we define the PIE as the mutual information between $U$ and $Y$: 
\begin{align}
\text{PIE} = I(U;Y) ~~ (bits).
\label{eq:PIE}
\end{align}
Note that traditional re-identification measures such as the re-identification rate \cite{Gambs_JCSS14,Mulder_WPES08,Murakami_TIFS17} specify an identification algorithm used by an adversary. 
However, even if the re-identification rate is low for the specific algorithm, 
the re-identification rate might be high for another algorithm used by the adversary. 
In contrast, the PIE does not specify the identification algorithm, and therefore is robust to the change of the identification algorithm. 
It is also robust to an unknown algorithm that may be used by the adversary in future. 

However, 
we specify a distribution 
$p_{U,X}$ 
to calculate $I(U;Y)$ in (\ref{eq:PIE}). 
When $I(U;Y)$ is close to $0$, almost no information about the user $U$ is obtained through 
the obfuscated data $Y$. 
Thus 
it is desirable to make $I(U;Y)$ small to prevent 
identification for any distribution 
$p_{U,X}$. 
Based on this, we define the notion called 
\textit{$(\calU,\alpha)$-PIE privacy:} 

\begin{definition}[$(\calU,\alpha)$-PIE privacy]
\label{def:PIE-privacy}
Let 
$\calU \subseteq \Omega$ 
and $\alpha \in \nngreals$. 
An obfuscation mechanism $\bmQ$ provides 
\emph{$(\calU,\alpha)$-PIE privacy} 
if 
\begin{align}
\sup_{p_{U,X}} I(U;Y) \leq \alpha ~~ (bits).
\label{eq:PIE-privacy}
\end{align}
\end{definition}
$(\calU,\alpha)$-PIE privacy guarantees that the PIE is upper-bounded by $\alpha$ for a user set $\calU$. 
Since the inequality (\ref{eq:PIE-privacy}) holds for any distribution $p_{U,X}$, 
the PIE is upper-bounded by $\alpha$ \textit{irrespective of the adversary's background knowledge}.
In other words, PIE privacy does not specify an identification algorithm nor the adversary's background knowledge, hence is robust to the change of the identification algorithm or the adversary's background knowledge.

Note that although $(\calU,\alpha)$-PIE privacy specifies a user set $\calU$, we show in Section~\ref{sec:theoretical} that LDP mechanisms $\bmQ$ such as the RR and GLH provide $(\calU,\alpha)$-PIE privacy for any 
$\calU$ with $|\calU| = n$, where $\alpha$ depends on $n$. 
Therefore, each user only has to know the number of users $n$ in the application to obfuscate her personal data with PIE guarantees, and does not have to know who else are using the application. 
We assume that the number of users $n$ is published by the data collector in advance.

The parameter $\alpha$ plays a role similar to the privacy budget $\epsilon$ in LDP. 
We explain how to set $\alpha$ in Section~\ref{sub:properties}.
We also show some basic properties of 
$(\calU,\alpha)$-PIE privacy 
in Section~\ref{sub:properties}. 
We 
show that 
$(\calU,\alpha)$-PIE privacy 
is a relaxation of $\epsilon$-LDP in Section~\ref{sub:LDP_PIE}. 

Cuff and Yu \cite{Cuff_CCS16} introduced the mutual information DP, which is a relaxation of DP using the mutual information. In the local privacy model, the notion in \cite{Cuff_CCS16} can be expressed as: $I(X;Y) \leq \alpha$. 
We call this notion \textit{MI-LDP (Mutual Information LDP)}. 
MI-LDP aims to prevent the inference of $X$, as with LDP. 
In contrast, PIE privacy aims to prevent the identification of $U$. 
This difference is significant. 
In fact, we show in Sections~\ref{sub:util_anal} and \ref{sub:implications} that 
our PIE privacy has a very different implication for utility and privacy than LDP. 

\smallskip
\noindent{\textbf{PSE.}}~~In Section~\ref{sec:theoretical}, we show that LDP mechanisms $\bmQ$ such as the RR and GLH provide PIE privacy, 
which 
upper-bounds the PIE. 
To evaluate how tight our upper-bounds 
are, 
we also 
introduce the \textit{PSE (Personal Identification System Entropy)}, which provides 
a \textit{lower-bound} on the PIE by specifying an 
identification algorithm. 

Specifically, we define the PSE as the mutual information between a user $U$ and a score vector $\bmS$: 
\begin{align}
\text{PSE} = I(U;\bmS) ~~ (bits).
\label{eq:PSE}
\end{align}
It 
is difficult to 
calculate 
the PIE in (\ref{eq:PIE}) through experiments, 
especially when $Y$ is in the high-dimensional feature space. 
On the other hand, 
the PSE can be easily calculated based on the theoretical results in \cite{Takahashi_IMAVIS14}. 
Therefore, we use $(\calU,\alpha)$-PIE privacy to guarantee that the PIE is less than or equal to $\alpha$ irrespective of the adversary's background knowledge, and use the PSE to evaluate how tight $\alpha$ is.
Note that a different identification algorithm leads to a different value of the PSE. 
In our experiments, we use an identification algorithm based on the Markov chain model \cite{Gambs_JCSS14,Mulder_WPES08,Murakami_TIFS17,Shokri_PETS11} to maximize the PSE (for details, see discussion after Proposition~\ref{prop:PIE_PSE}).

Below we explain how to calculate the PSE. 
Let $f_G$ (resp.~$f_I$) be a one-dimensional distribution of scores (similarities or distances) 
output by the matching algorithm $\calF$ 
when the obfuscated data and the auxiliary information belong to 
the same user (resp.~different users). 
We refer to $f_G$ as a genuine score distribution, and $f_I$ as an impostor score distribution in the same way as biometric literature \cite{guide,intro}. 

Then 
$I(U;\bmS)$ converges to the KL divergence \cite{elements} between $f_G$ and $f_I$ as $n$ increases: 
\begin{theorem}[Theorem 6 in \cite{Takahashi_IMAVIS14}]
\label{thm:BSE_KL} 
Let $f_G$ be a genuine score distribution and $f_I$ be an impostor score distribution (both $f_G$ and $f_I$ are one-dimensional). Then
\begin{align}
I(U;\bmS) \rightarrow D(f_G || f_I) \hspace{3mm} (n \rightarrow \infty).
\label{eq:BSE_KL}
\end{align}
\end{theorem}
Theorem~\ref{thm:BSE_KL} is proved 
in \cite{Takahashi_IMAVIS14} when 
$Y$ is a biometric feature. 
Theorem~\ref{thm:BSE_KL} means that the PSE converges to the KL divergence $D(f_G||f_I)$ as the number of users $n$ increases. 
In our experiments, 
we 
estimate 
$D(f_G||f_I)$ by the generalized $k$-NN estimator \cite{Wang_TIT09}, which is asymptotically unbiased; i.e., it converges to the true value as the sample size increases.

The following proposition shows the relation between the PIE and PSE:

\begin{proposition}[PIE and PSE]
\label{prop:PIE_PSE} 
For any matching algorithm $\calF$ and any auxiliary information $r_1, \cdots, r_n$, 
\begin{align}
I(U;\bmS) \leq I(U;Y)
\label{eq:PIE_PSE}
\end{align}
with equality if and only if $U$, $\bmS$, and $Y$ 
form 
the Markov chain $U \rightarrow \bmS \rightarrow Y$ (i.e., 
$\bmS$ is a sufficient statistic for $U$). 
\end{proposition}

\begin{proof}

\begin{figure}
\centering
\includegraphics[width=0.25\linewidth]{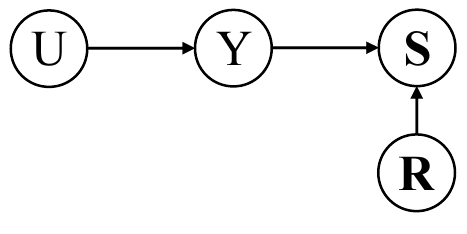}
\vspace{-2mm}
\caption{Graphical model of user $U$, obfuscated data $Y$, auxiliary information $\bmR$, and score vector $\bmS$. 
}
\label{fig:graphical}
\end{figure}

Let $\bmR$ be a random variable representing the auxiliary information $(r_1, \cdots, r_n)$. 
By Figures~\ref{fig:framework_obf} and \ref{fig:framework_ide}, 
$U$, $Y$, $\bmR$, and $\bmS$ are represented as a graphical model in Figure~\ref{fig:graphical}. 
Note that 
$U$ (unknown user ID) is generated from 
$p_U$ 
independently of $\bmR$, $X$ is generated from 
$p_{X|U=u_i}$, 
and then $Y$ is generated using $\bmQ$ (as described in Section~\ref{sub:framework}). 
Therefore, $\bmR$ is independent of both $U$ and $Y$. 
$Y$ is called \textit{head-to-tail} with respect to the path from $U$ to $\bmS$ \cite{prml}. 
After we observe $Y$, a path from $U$ to $\bmS$ is blocked. 
Since the head-to-tail node blocks the path, 
$U$ and $\bmS$ are conditionally independent given $Y$ \cite{prml}; i.e., $U$, $Y$, and $\bmS$ form the Markov chain $U \rightarrow Y \rightarrow \bmS$.

Then 
by the data processing inequality \cite{elements}, 
which states that post-processing cannot increase information, 
(\ref{eq:PIE_PSE}) holds with equality if and only if $U \rightarrow \bmS \rightarrow Y$. 
\end{proof}

Proposition~\ref{prop:PIE_PSE} states that the PIE can be lower-bounded by the PSE. 
In addition, \textit{the PSE is equal to the PIE} in some cases. 
For example, let $q_i$ be a distribution of obfuscated data $y \in \calY$ for user $u_i \in \calU$; i.e., $q_i(y)$ is the likelihood that $y$ belongs to $u_i$. 
Then the equality in (\ref{eq:PIE_PSE}) holds 
if 
the adversary uses the likelihood $q_i(y)$ as a similarity for user $u_i$; i.e., $s_i = q_i(y)$. 
(In this case, $q_i(y)$ depends on $u_i$ only through $s_i$. Then by the Fisher-Neyman factorization theorem \cite{Lehmann06}, $U \rightarrow \bmS \rightarrow Y$.) 
In other words, if the adversary knows a generative model of $Y$ given $U=u_i$, she can use it to achieve PSE $=$ PIE; otherwise, the PSE may be smaller than the PIE.

In Section~\ref{sec:theoretical}, we prove upper-bounds on the PIE for the RR and GLH. 
Then in our experiments, we estimate $q_i(y)$ by assuming the Markov chain model \cite{Gambs_JCSS14,Mulder_WPES08,Murakami_TIFS17,Shokri_PETS11}, and used it as $s_i$. 
We show that the PSE for this adversary is close to our upper-bounds on the PIE.

\smallskip
\noindent{\textbf{BSE.}}~~The PSE includes a 
measure 
of identifiability for biometric data as a special case. 
Specifically, the BSE (Biometric System Entropy) is defined in \cite{Takahashi_IMAVIS14} as the mutual information $I(U;\bmS)$ in the case where $Y$ is a biometric feature.
Thus, the BSE is a special case of the PSE in the case where $Y$ in (\ref{eq:PSE}) is a biometric feature. 

Based on this, we compare the identifiability of personal data (e.g., location data, rating history) with that of biometric data (e.g., fingerprint, face) using the PSE in our experiments.

\subsection{Basic Properties of the PIE}
\label{sub:properties}
Here we show several basic properties of the PIE (or PIE privacy). 

\smallskip
\noindent{\textbf{Privacy Axioms.}}~~Kifer and Lin \cite{Kifer_JPC12} propose that any privacy definition should satisfy two privacy axioms: \textit{post-processing invariance} and \textit{convexity}. 
We first show that 
$(\calU,\alpha)$-PIE privacy 
in Definition~\ref{def:PIE-privacy} provides these privacy axioms: 

\begin{theorem}[Post-processing invariance]
\label{thm:post-processing} 
Let 
$n \in \nats$ 
and $\alpha \in \nngreals$. 
Let $\bmQ$ an obfuscation mechanism. 
Let $\lambda$ be a randomized algorithm whose input space contains the output space of $\bmQ$ 
and whose randomness is independent of both the personal data and the randomness in $\bmQ$. 
If 
the 
obfuscation mechanism $\bmQ$ 
provides 
$(\calU,\alpha)$-PIE privacy, 
then $\lambda \circ \bmQ$ provides 
$(\calU,\alpha)$-PIE privacy. 
\end{theorem}

\begin{theorem}[Convexity]
\label{thm:convexity} 
Let 
$n \in \nats$ 
and $\alpha \in \nngreals$. 
Let $\bmQ_1$ and $\bmQ_2$ be obfuscation mechanisms that provide 
$(\calU,\alpha)$-PIE privacy. 
For any $w\in[0,1]$, let $\bmQ_w$ an obfuscation mechanism that executes $\bmQ_1$ with probability $w$ and $\bmQ_2$ with probability $1-w$. 
Then $\bmQ_w$ provides 
$(\calU,\alpha)$-PIE privacy. 
\end{theorem}

The proofs are given in Appendix~\ref{sec:proof_sec4}. 
Theorem~\ref{thm:post-processing} 
guarantees that post-processing, which can be performed by a data collector or an adversary, cannot break 
$(\calU,\alpha)$-PIE privacy. 
Theorem~\ref{thm:convexity} 
allows users to randomly choose which obfuscation mechanism to use. 

\smallskip
\noindent{\textbf{Identification Error.}}~~We can use PIE privacy to lower-bound the lowest possible re-identification error probability by specifying the prior distribution $p_U$. 
Specifically, assume an adversary who 
has a prior distribution $p_U$ and 
attempts to identify a user $U$ based on a score vector $\bmS$. 
Let $\beta_U\in[0,1]$ and $\beta_{U|\bmS}\in[0,1]$ be the following probabilities:
\begin{align*}
\beta_U &= 1 - \max_{u_i\in\calU} \Pr(U=u_i) \nonumber\\
\beta_{U|\bmS} &= 1 - \bbE_\bmS \left[\max_{u_i\in\calU} \Pr(U=u_i|\bmS) \right],
\end{align*}
where for $a\in\reals$, $\bbE_\bmS[a]$ represents the expectation of $a$ over $\bmS$. 
$\beta_U$ (resp.~$\beta_{U|\bmS}$) is the \textit{Bayes error probability} before (resp.~after) observing $\bmS$. 
In other words, $\beta_U$ and $\beta_{U|\bmS}$ are the lowest possible identification error probabilities for any classifier. 
Then $\beta_{U|\bmS}$ is lower-bounded by 
the PSE and PIE: 

\begin{proposition}[Bayes error, PSE, and PIE]
\label{prop:identification_error} 
\begin{align}
\beta_{U|\bmS} &\geq 1 + \frac{I(U;\bmS)+1}{\log(1-\beta_U)} \geq 1 + \frac{I(U;Y)+1}{\log(1-\beta_U)}.
\label{eq:Fano_general}
\end{align}
If $U$ is uniformly distributed, then 
\begin{align}
\beta_{U|\bmS} &\geq 1 - \frac{I(U;\bmS)+1}{\log n} \geq 1 - \frac{I(U;Y)+1}{\log n}.
\label{eq:Fano_uniform}
\end{align}
\end{proposition}

\begin{proof}
The first inequality in (\ref{eq:Fano_general}) is the generalized Fano's inequality by Han and Verd\'{u} \cite{Han_TIT94}. 
By (\ref{eq:PIE_PSE}), the second inequality in (\ref{eq:Fano_general}) holds (note that $\log(1-\beta_U)$ is negative). 
If 
$U$ is uniform, 
then $\beta_U=1-\frac{1}{n}$. 
Therefore, 
(\ref{eq:Fano_uniform}) holds.
\end{proof}

\begin{corollary}[Bayes error and PIE privacy]
\label{cor:identification_error_cor} 
Let 
$\calU \subseteq \Omega$ 
and $\alpha \in \nngreals$. 
If an obfuscation mechanism $\bmQ$ provides 
$(\calU,\alpha)$-PIE privacy, 
then 
\begin{align}
\beta_{U|\bmS} \geq 1 - \frac{\alpha+1}{\log(1-\beta_U)}.
\label{eq:Fano_general_cor}
\end{align}
If $U$ is uniformly distributed, then 
\begin{align}
\beta_{U|\bmS} \geq 1 - \frac{\alpha+1}{\log n}.
\label{eq:Fano_uniform_cor}
\end{align}
\end{corollary}
Corollary~\ref{cor:identification_error_cor} is immediately derived from Definition~\ref{def:PIE-privacy} and Proposition~\ref{prop:identification_error}. 
Corollary~\ref{cor:identification_error_cor} means that 
$(\calU,\alpha)$-PIE privacy 
provides an \textit{absolute} guarantee about the posterior error probability as a function of $\alpha$ and the prior error probability. 
Given a required identification error probability, we can use %Proposition~\ref{prop:identification_error} 
Corollary~\ref{cor:identification_error_cor} to derive $\alpha$ 
that satisfies the requirement. 

For example, assume that 
there are 
$n=10^6$ 
users in a certain application 
and each user sends a single personal datum. 
In this case, we can assume that $p_U$ is uniform, as described in Section~\ref{sub:framework}. 
If we require $\beta_{U|\bmS} > 0.8$ (resp.~$0.5$), then we should set $\alpha<2.07$ (resp.~$6.21$). 
In our experiments, we also evaluate how tight the lower-bounds in Proposition~\ref{prop:identification_error} (i.e., (\ref{eq:Fano_general}) and (\ref{eq:Fano_uniform})) are.

\smallskip
\noindent{\textbf{Remark.}}~~As described in Section~\ref{sub:PIE}, PIE privacy is a privacy metric that does not depend on the adversary's background knowledge. 
On the other hand, (\ref{eq:Fano_general_cor}) in Corollary~\ref{cor:identification_error_cor} depends on the prior Bayes error probability $\beta_U$, hence depends on the prior distribution $p_U$ of the adversary. 

Nevertheless, we emphasize that Corollary~\ref{cor:identification_error_cor} is useful because we can make a reasonable assumption about the prior distribution $p_U$ in many practical applications; e.g., if each user sends a single personal datum, then the prior distribution $p_U$ is uniform, as described in Section~\ref{sub:framework}.

\smallskip
\noindent{\textbf{PIE and Obfuscation Mechanism.}}~~Finally, we show the relation between the PIE and the obfuscation mechanism $\bmQ$. 
The PIE is a 
measure 
of leakage between $U$ and $Y$, whereas $\bmQ$ transforms $X$ into $Y$. 
The following proposition provides a simple guideline on when to use $\bmQ$ to prevent re-identification:

\begin{proposition}[PIE and obfuscation mechanism]
\label{prop:PIE_obf} 
\begin{align}
I(U;Y) \leq \min\{I(U;X), I(X;Y)\}.
\label{eq:PIE_obf}
\end{align}
\end{proposition}

\begin{proof}
$U$, $X$, and $Y$ form the Markov chain $U \rightarrow X \rightarrow Y$. 
Then 
by the data processing inequality \cite{elements}, (\ref{eq:PIE_obf}) holds. 
\end{proof}

Proposition~\ref{prop:PIE_obf} states that 
we do not need to use $\bmQ$ when $I(U;X)$ is small. 
For example, 
assume that there are $n=10^8$ users in 
a certain application. 
Each user sends a single (possibly obfuscated) personal datum; i.e., $U$ is uniform for the adversary. 
Here the personal datum is 
his/her income (dollars) discretized as in \cite{USCensus}: $0$, $[1,14999]$, $[15000,29999]$, $[30000,60000]$, or $[60000,\infty)$ (i.e., five categories). 
The data collector \allowbreak pseudonymizes (randomly permutates) the 
personal or obfuscated data of all users, as described in 
Section~\ref{sec:intro}. 
Consider the risk of re-identification. 

In this example, even if $\bmQ$ is not used (i.e., $Y=X$), 
$I(U;Y)=I(U;X) \leq H(X) \leq \log 5$. 
Then by (\ref{eq:Fano_uniform}), $\beta_{U|\bmS} \geq 0.88$. 
This means that \textit{pseudonymization (random permutation) alone} achieves 
$(\calU,\log 5)$-PIE privacy 
and 
fairly prevents 
the risk of re-identification in this case. 
More generally, 
pseudonymization alone prevents re-identification when 
$|\calX|$ is small. 
Although we assume that $U$ is uniform in the above example, we can also make a similar argument even when $U$ is non-uniform. 
For example, if $n=10^8$, $|\calX|=5$, and the maximum of $\Pr(U=u_i)$ over $\calU$ is $0.01$ (i.e., $10^6$ times larger than the average), then pseudonymization alone achieves $\beta_{U|\bmS} \geq 0.5$ (by (\ref{eq:Fano_general})).

This example illustrates the difference between PIE privacy (user privacy) and LDP (data privacy). 
Recent studies \cite{Balle_arXiv19,Erlingsson_SODA19} have shown that 
$\epsilon$ in LDP can be significantly reduced by introducing a server (shuffler) that randomly permutes all obfuscated data. 
However, when each user does not obfuscate her personal data (i.e., when $\epsilon = \infty$ at the client side), then this technique does not provide 
DP 
($\epsilon$ is still $\infty$ after the random permutation in \cite{Balle_arXiv19,Erlingsson_SODA19}). 
It follows from \cite{Cuff_CCS16} that when $\epsilon$ in DP is $\infty$, $\epsilon$ in its average versions (i.e., KL-DP \cite{Barber_arXiv14,Cuff_CCS16}, mutual information DP \cite{Cuff_CCS16}) is also $\infty$. 
On the other hand, our PIE privacy guarantees a high re-identification error probability even in this case, given that $|\calX|$ is small.

Proposition~\ref{prop:PIE_obf} shows that 
when $|\calX|$ is large (and so is $I(U;X)$), we should use $\bmQ$ to prevent re-identification. 
In Section~\ref{sec:theoretical}, we show how much the PIE can be reduced by using LDP mechanisms $\bmQ$ (e.g., RR, GLH).

\section{Theoretical Analysis}
\label{sec:theoretical}
We now investigate the effectiveness of 
$\epsilon$-LDP 
mechanisms $\bmQ$ in terms of 
$(\calU,\alpha)$-PIE privacy 
and utility. 
Section~\ref{sub:LDP_PIE} shows a general bound on $\alpha$, which holds for any $\epsilon$-LDP mechanism. 
Sections~\ref{sub:RR_PIE} and \ref{sub:GLH_PIE} show that much tighter bounds can be obtained for the $\epsilon$-RR and $(g,\epsilon)$-GLH, respectively. 
Section~\ref{sub:util_anal} 
analyzes the utility of the $\epsilon$-RR and $(g,\epsilon)$-GLH 
for given PIE guarantees. 
Section~\ref{sub:implications} discusses implications for privacy and utility based on our theoretical results. 

The proofs of all statements in this section are given in Appendix~\ref{sec:proof_sec5}.

\subsection{LDP and PIE}
\label{sub:LDP_PIE}
We first show the relation between LDP and PIE privacy:

\begin{proposition}[LDP and PIE]
\label{prop:LDP_PIE} 
If an obfuscation mechanism $\bmQ$ provides $\epsilon$-LDP, then it provides 
$(\calU,\alpha)$-PIE privacy 
for any $\calU \subseteq \Omega$ such that $|\calU|=n$, 
where
\begin{align}
\alpha = \min\{\epsilon \log e, \epsilon^2 \log e, \log n, \log |\calX|\}.
\label{eq:alpha_LDP}
\end{align}
\end{proposition}

Proposition~\ref{prop:LDP_PIE} is 
derived from Proposition~\ref{prop:PIE_obf}. 
Specifically, 
we use Lemma~1 in \cite{Cuff_CCS16}, which states that $\bmQ$ satisfies $I(X;Y) \leq \min\{\epsilon \log e,\epsilon^2 \log e\}$ (bits), 
and the fact that 
$I(U;X) \leq \min\{\log n, \log |\calX|\}$. See Appendix~\ref{sec:proof_sec5} for details.

Proposition~\ref{prop:LDP_PIE} means that 
$(\calU,\alpha)$-PIE privacy 
is a relaxation of $\epsilon$-LDP; i.e., any LDP mechanism provides PIE privacy. 
This is a general result 
that holds for any LDP mechanism. 
However, the bound in (\ref{eq:alpha_LDP}) is loose, as we will show below. 

\subsection{PIE of the RR}
\label{sub:RR_PIE}
If we consider a specific LDP mechanism, a much tighter bound on the PIE can be obtained. 
We first show such a bound for the $\epsilon$-RR in Section~\ref{sub:obf}. 

\begin{theorem}[PIE of the RR]
\label{thm:PIE_RR} 
Let $\theta_{RR} = \frac{e^\epsilon-1}{|\calX|+e^\epsilon-1}$. 
For any distribution 
$p_{U,X}$, 
the $\epsilon$-RR satisfies
\begin{align}
I(U;Y) \leq \theta_{RR} I(U;X).
\label{eq:RR_IUY_IUX}
\end{align}
Therefore, the $\epsilon$-RR provides 
$(\calU,\alpha)$-PIE privacy for any $\calU \subseteq \Omega$ such that $|\calU|=n$, 
where
\begin{align}
\alpha = \theta_{RR} \min\{\log n, \log |\calX|\}.
\label{eq:PIE_RR}
\end{align}
\end{theorem}
By (\ref{eq:RR_IUY_IUX}), the $\epsilon$-RR can decrease the mutual information between a user and personal data by $\theta_{RR} = \frac{e^\epsilon-1}{|\calX|+e^\epsilon-1}$. 
Thus we refer to $\theta_{RR}$ as the \textit{MI (Mutual Information) loss parameter}. 

\begin{figure}
\centering
\includegraphics[width=0.75\linewidth]{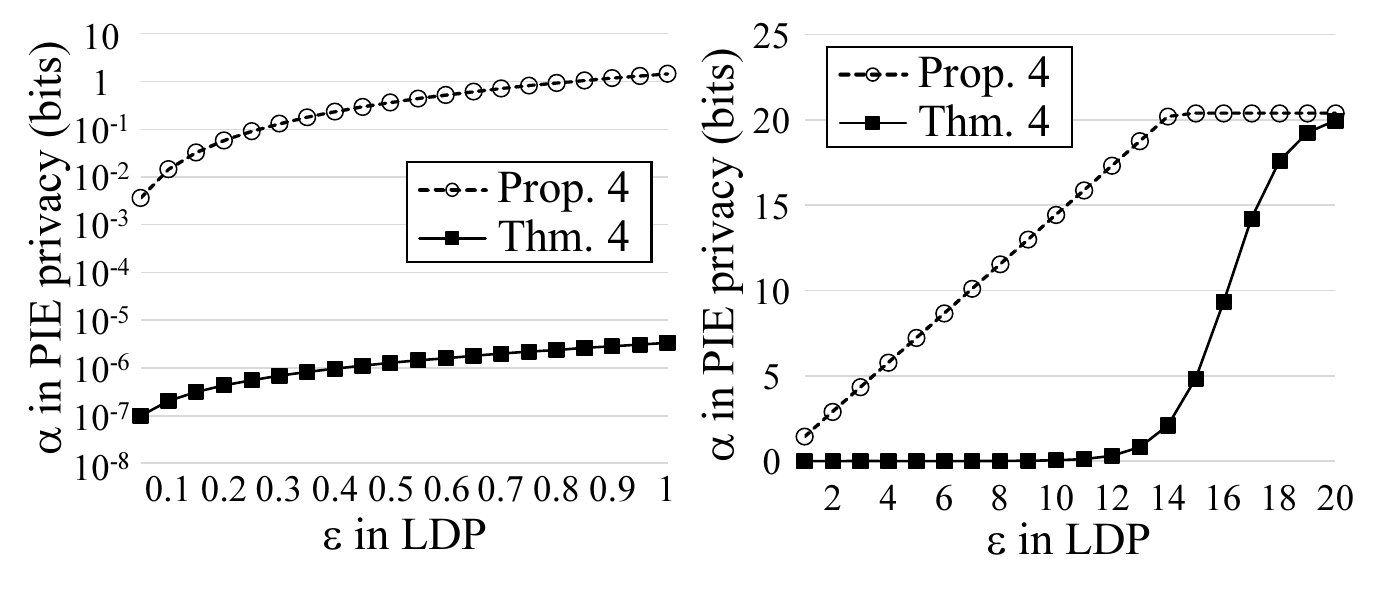}
\vspace{-2mm}
\caption{Relation between $\epsilon$ in LDP and $\alpha$ in PIE privacy ($n=1370637$, $|\calX|=10500393$).}
\label{fig:PIE_RR}
\end{figure}

Figure~\ref{fig:PIE_RR} shows the relation between $\epsilon$ in LDP and $\alpha$ in PIE privacy. 
Here we set $n=1370637$ and $|\calX|=10500393$ because these values were used in our experiments using 
the location data. 
Figure~\ref{fig:PIE_RR} shows that $\alpha$ in Theorem~\ref{thm:PIE_RR} is much smaller than $\alpha$ in Proposition~\ref{prop:LDP_PIE} for both small $\epsilon$ ($0.05$ to $1$) and large $\epsilon$ ($1$ to $20$). 
For example, when $\epsilon=0.1$, $1$, and $10$, $\alpha$ in Proposition~\ref{prop:LDP_PIE} is 
$\alpha=0.014$, $1.4$, and $14$, 
respectively, whereas $\alpha$ in Theorem~\ref{thm:PIE_RR} is 
$\alpha=2.0\times10^{-7}$, $3.3\times10^{-6}$, and $0.043$, 
respectively. 
As 
the input alphabet size $|\calX|$ 
increases, the difference between two $\alpha$ values becomes larger because the MI loss parameter $\theta_{RR}$ decreases with increase in $|\calX|$. 
In other words, $\alpha$ in Theorem~\ref{thm:PIE_RR} is much tighter, especially when $|\calX|$ is large. 
In our experiments, we also show that $\alpha$ in Theorem~\ref{thm:PIE_RR} is close to the PSE for a specific 
identification algorithm. 

We also show that the $\epsilon$-RR has compositionality  \cite{DP} 
in that 
$\alpha$ in (\ref{eq:PIE_RR}) increases linearly for multiple personal data: 

\begin{theorem}[Composition of the RR]
\label{thm:Comp_RR} 
Let 
$\theta_{RR} = \frac{e^\epsilon-1}{|\calX|+e^\epsilon-1}$. 
Assume that a user $U$ has $t\in\nats$ personal data and obfuscates each of the data by using the $\epsilon$-RR. 
For $i\in[t]$, let $X^{(i)}$ and $Y^{(i)}$ be random variables representing the $i$-th personal data and obfuscated data, respectively. 
Then the combined release $Y=(Y^{(1)},\cdots,Y^{(t)})$ 
provides 
$(\calU,t\alpha)$-PIE privacy, 
where $\alpha = \theta_{RR} \min\{\log n, \allowbreak \log |\calX|\}$.
\end{theorem}
Theorem~\ref{thm:Comp_RR} upper-bounds the PIE when a user has multiple personal data and obfuscates each of them using the $\epsilon$-RR. 
Note that Theorem~\ref{thm:Comp_RR} holds 
irrespective of whether there are correlations between $X^{(1)}, \cdots, X^{(t)}$ (see 
Appendix~\ref{sec:proof_sec5} for details).

\subsection{PIE of the GLH}
\label{sub:GLH_PIE}
Next we show a tighter bound on the PIE for the $(g,\epsilon)$-GLH in Section~\ref{sub:obf}:

\begin{theorem}[PIE of the GLH]
\label{thm:PIE_GLH} 
Let $\theta_{GLH} = \frac{e^\epsilon-1}{g+e^\epsilon-1}$. 
For any distribution 
$p_{U,X}$, 
the $(g,\epsilon)$-GLH satisfies
\begin{align}
I(U;Y) \leq \theta_{GLH} I(U;X).
\label{eq:GLH_IUY_IUX}
\end{align}
Therefore, the $(g,\epsilon)$-GLH provides 
$(\calU,\alpha)$-PIE privacy for any $\calU \subseteq \Omega$ such that $|\calU|=n$, 
where
\begin{align}
\alpha = \theta_{GLH} \min\{\log n, \log |\calX|\}.
\label{eq:PIE_GLH}
\end{align}
\end{theorem}
Theorem~\ref{thm:PIE_GLH} can be proved in an analogous way to Theorem~\ref{thm:PIE_RR} because the $(g,\epsilon)$-GLH uses the RR after applying random projection to $x\in\calX$, as shown in (\ref{eq:GLH}). 

By (\ref{eq:GLH_IUY_IUX}), the $(g,\epsilon)$-GLH can decrease the mutual information between a user and personal data by $\theta_{GLH}$. 
As with $\theta_{RR}$, we call $\theta_{GLH}$ the \textit{MI loss parameter}. 
$\alpha$ in Theorem~\ref{thm:PIE_GLH} is much tighter than $\alpha$ in Proposition~\ref{prop:LDP_PIE}, especially when $g$ is large. 

As with the $\epsilon$-RR, the $(g,\epsilon)$-GLH also composes (irrespective of whether there are correlations between $X^{(1)}, \cdots, X^{(t)}$): 
\begin{theorem}[Composition of the GLH]
\label{thm:Comp_GLH} 
Let $\theta_{GLH} = \frac{e^\epsilon-1}{g+e^\epsilon-1}$. 
Assume that a user $U$ has $t\in\nats$ personal data and obfuscates each of the data by using the $(g,\epsilon)$-GLH. 
For $i\in[t]$, let $X^{(i)}$ and $Y^{(i)}$ be random variables representing the $i$-th personal data and obfuscated data, respectively. 
Then the combined release $Y=(Y^{(1)},\cdots,Y^{(t)})$ provides 
$(\calU,t\alpha)$-PIE privacy, 
where $\alpha = \theta_{GLH} \min\{\log n, \allowbreak \log |\calX|\}$.
\end{theorem}

\smallskip
\noindent{\textbf{Relation between the MI Loss Parameter and the Bayes Error Probability.}}~~In 
both the RR and GLH, 
we can set the MI loss parameter to achieve a required identification error probability. 

For example, assume that each user sends a single personal datum (i.e., $U$ is uniform) and that we require the Bayes error probability to be $\beta_{U|\bmS} > 0.5$. 
By Proposition~\ref{prop:identification_error} and Theorem~\ref{thm:PIE_RR} (resp.~Theorem~\ref{thm:PIE_GLH}), this requirement is satisfied by setting the MI parameter to $\theta_{RR} = 0.5$ (resp.~$\theta_{GLH} = 0.5$) when $n$ is large.

\subsection{Utility Analysis}
\label{sub:util_anal}
Finally, we analyze the utility of the $\epsilon$-RR and the $(g,\epsilon)$-GLH. 
Here as with \cite{Kairouz_ICML16,Wang_USENIX17,Murakami_USENIX19}, we assume that each user sends a single datum and consider distribution estimation 
as a task for a data collector.

\smallskip
\noindent{\textbf{Preliminaries.}}~~For $i\in[n]$, let $X_i$ and $Y_i$ be random variables representing user $u_i$'s personal data and obfuscated data, respectively. 
Let $\bmX_{1:n}= (X_1, \cdots, X_n)$ and $\bmY_{1:n} = (Y_1, \cdots, Y_n)$. 
Let 
$\calC$ be the probability simplex, and $\bmp\in\calC$ be an empirical distribution of $\bmX_{1:n}$, whose probability of $x\in\calX$ is given by $\bmp(x)$. 
In distribution estimation, 
the data collector estimates $\bmp$ from $\bmY_{1:n}$. 

For estimating $\bmp$ 
from $\bmY_{1:n}$, we use an empirical estimator in the same way as \cite{Agrawal_SIGMOD05,Huang_ICDE08,Kairouz_ICML16,Wang_USENIX17,Murakami_USENIX19} because it is easy to analyze. 
Let $\hbmp_{RR}$ be an empirical estimate of $\bmp$ when the $\epsilon$-RR is used. 
Simple calculations from (\ref{eq:RR}) show that: 
\begin{align}
\hbmp_{RR}(x) = \frac{1}{\mu_{RR}-\nu_{RR}} \left(\frac{\bmc_{RR}(x)}{n} - \nu_{RR}\right), 
\label{eq:RR_emp}
\end{align}
where $\mu_{RR} = \frac{e^\epsilon}{|\calX|+e^\epsilon-1}$, $\nu_{RR} = \frac{1}{|\calX|+e^\epsilon-1}$, and $\bmc_{RR}(x)$ is the number of $x\in\calX$ in $\bmY_{1:n}$ (note that $\calX=\calY$ in the RR). 

Similarly, let $\hbmp_{GLH}$ be an empirical estimate of $\bmp$ when the $(g,\epsilon)$-GLH is used. 
Let $\calY_x = \{(h,y)\in\calY | y=h(x)\} \subseteq \calY$. 
Then the empirical estimate $\hbmp_{GLH}$ can be written as follows \cite{Wang_USENIX17}:
\begin{align}
\hbmp_{GLH}(x) = \frac{1}{\mu_{GLH} - \nu_{GLH}} \left(\frac{\bmc_{GLH}(x)}{n} - \nu_{GLH} \right),
\label{eq:GLH_emp}
\end{align}
where $\mu_{GLH}=\frac{e^\epsilon}{g+e^\epsilon-1}$, $\nu_{GLH}=\frac{1}{g}$, and $\bmc_{GLH}(x)$ is the number of $(h,y)\in\calY_x$ in $\bmY_{1:n}$. 

As with \cite{Kairouz_ICML16,Wang_USENIX17,Murakami_USENIX19}, we use the expected $l_2$ loss between the true probability and the estimate as a utility measure. 
Specifically, we fix $\bmX_{1:n}$ (and hence $\bmp$). 
Given the estimate $\hbmp$ of $\bmp$, we evaluate: 
\begin{align}
\bbE[(\bmp(x) - \hbmp(x))^2]
\label{eq:l2_loss}
\end{align}
for each input 
symbol 
$x\in\calX$, where the expectation is taken over all possible realizations of $\bmY_{1:n}$ (hereinafter we omit the subscript $\bmY_{1:n}$ in $\bbE$). 
Note that when the empirical 
estimator 
is used, the expected $l_2$ loss is equal to the variance of $\hbmp(x)$ because the empirical estimate is unbiased \cite{Kairouz_ICML16,Wang_USENIX17}. 

\smallskip
\noindent{\textbf{Expected $l_2$ Loss.}}~~We 
first show 
the expected $l_2$ loss for the $\epsilon$-RR and the $(g,\epsilon)$-GLH: 

\begin{proposition}[$l_2$ loss of the RR]
\label{prop:l2_RR} 
For any $x\in\calX$, 
\begin{align}
\bbE[(\bmp(x) - \hbmp_{RR}(x))^2] = \frac{|\calX|+e^\epsilon-2}{n(e^\epsilon-1)^2} + \frac{\bmp(x)(|\calX|-2)}{n(e^\epsilon-1)}.
\label{eq:l2_RR}
\end{align}
\end{proposition}

\begin{proposition}[$l_2$ loss of the GLH]
\label{prop:l2_GLH} 
For any $x\in\calX$, 
\begin{align}
&\bbE[(\bmp(x) - \hbmp_{GLH}(x))^2] \nonumber\\
&= \frac{(g+e^\epsilon-1)^2}{n(e^\epsilon-1)^2(g-1)} + \frac{\bmp(x)(g^2-2g-e^\epsilon+1)}{n(e^\epsilon-1)(g-1)}.
\label{eq:l2_GLH}
\end{align}
\end{proposition}
The first terms in (\ref{eq:l2_RR}) and (\ref{eq:l2_GLH}) are shown in \cite{Wang_USENIX17} for the task of estimating counts of $x\in\calX$ in $\calX_{1:n}$ (we need to multiply the values in \cite{Wang_USENIX17} by $1/n^2$ to normalize counts to probabilities). 
The second terms in (\ref{eq:l2_RR}) and (\ref{eq:l2_GLH}) are obtained by simple calculations. 

Next we show an optimal parameter $g$ in the $(g,\epsilon)$-GLH while fixing the bound on $\alpha$. 
Specifically, 
by Theorem~\ref{thm:PIE_GLH}, 
the MI loss parameter $\theta_{GLH}$ ($= \frac{e^\epsilon-1}{g+e^\epsilon-1}$) 
determines the bound on $\alpha$ for fixed $n$ and $|\calX|$. 
Therefore, we fix $\theta_{GLH}$ and find the optimal $g$ that minimizes the expected $l_2$ loss. 
Note that for a fixed $\theta_{GLH}$ ($= \frac{e^\epsilon-1}{g+e^\epsilon-1}$), $\epsilon$ increases with increase in $g$.

For a fixed privacy budget $\epsilon$ in LDP, the optimal $g$ that minimizes the expected $l_2$ loss is given by: $g=e^\epsilon+1$ \cite{Wang_USENIX17}. 
In contrast, we show that for a fixed $\theta_{GLH}$, 
a larger $g$ provides better utility: 

\begin{theorem}[Optimal $g$ in the GLH]
\label{thm:optimal_g} 
For a fixed $\theta_{GLH}$ ($= \frac{e^\epsilon-1}{g+e^\epsilon-1}$), the expected $l_2$ loss of the $(g,\epsilon)$-GLH in (\ref{eq:l2_GLH}) is monotonically decreasing in $g$ and 
\begin{align}
\bbE[(\bmp(x) - \hbmp_{GLH}(x))^2] \rightarrow \frac{\bmp(x)(1-\theta_{GLH})}{n\theta_{GLH}} \hspace{3mm} (g \rightarrow \infty).
\label{eq:optimal_g}
\end{align}
\end{theorem}
Recall that $g$ in the $(g,\epsilon)$-GLH is the output size of the hash function $h: \calX \rightarrow [g]$. 
A larger $g$ preserves more information about the personal data $x\in\calX$. 
Therefore, Theorem~\ref{thm:optimal_g} is intuitive in that 
\textit{compressing the personal data $x$ with a smaller $g$ results in the loss of information about $x$, and hence causes the loss of utility}. 

\smallskip
\noindent{\textbf{RR vs. GLH.}}~~We compare the expected $l_2$ loss of the RR with that of the GLH. 
Here we set the MI loss parameters $\theta_{RR}$ and $\theta_{GLH}$ to the same value ($\theta_{RR} = \theta_{GLH}$) so that the bounds on $\alpha$ in (\ref{eq:PIE_RR}) and (\ref{eq:PIE_GLH}) are the same. 
We assume that $g$ in the GLH is very large and compare the right side of (\ref{eq:l2_RR}) with the right side of (\ref{eq:optimal_g}). 

When $\theta_{RR} = \theta_{GLH}$, the right side of (\ref{eq:optimal_g}) can be written as: $\frac{\bmp(x)(1-\theta_{GLH})}{n\theta_{GLH}} = \frac{\bmp(x)(1-\theta_{RR})}{n\theta_{RR}} = \frac{\bmp(x)|\calX|}{n(e^\epsilon-1)}$, where $\epsilon$ satisfies $\theta_{RR} = \frac{e^\epsilon-1}{|\calX|+e^\epsilon-1}$. 
Then for a large $|\calX|$, we obtain:
\begin{align*}
    \text{the right side of (\ref{eq:optimal_g})} \approx \text{the second term in (\ref{eq:l2_RR})}.
\end{align*}
Therefore, the remaining question is how large the first term in (\ref{eq:l2_RR}) is compared to the second term in (\ref{eq:l2_RR}).

As an example, we consider the case where 
$\theta_{RR}=\theta_{GLH} = 0.5$ 
(which guarantees 
the Bayes error probability $\beta_{U|\bmS}$ 
to be 
larger than 
$0.5$). 
In this case, simple calculations show that the first and second terms in (\ref{eq:l2_RR}) are almost equal to $\frac{2}{n|\calX|}$ and $\frac{\bmp(x)}{n}$, respectively. 
For a popular 
symbol 
$x\in\calX$ with a large probability $\bmp(x) \gg \frac{1}{|\calX|}$, we obtain $\frac{2}{n|\calX|} \ll \frac{\bmp(x)}{n}$. 
Therefore, the RR and GLH have almost the same utility for estimating the probabilities for popular 
symbols 
$x\in\calX$; i.e., heavy hitters \cite{Bassily_STOC15,Chan_PETS12,Erlingsson_CCS14,Hsu_ICALP12,Qin_CCS16}. 

For an unpopular 
symbol 
$x\in\calX$ with a small probability $\bmp(x) \ll \frac{1}{|\calX|}$, we obtain  $\frac{2}{n|\calX|} \gg \frac{\bmp(x)}{n}$. 
This means that the RR has much larger $l_2$ loss for the unpopular 
symbol. 
This is 
caused by the fact 
negative values are assigned to many unpopular 
symbols 
in the RR \cite{Agrawal_SIGMOD05}. 
However, zero (or very small positive) values can be assigned to the estimates 
$\hbmp(x)$ below a significance threshold (determined via Bon-ferroni correction) \cite{Erlingsson_CCS14}. 
Then both the RR and GLH would have small $l_2$ losses for unpopular 
symbols. 

In our experiments, we show that the RR and GLH have almost the same utility for both popular and unpopular 
symbols 
when we 
use the significant threshold. 
We also note that the communication cost of the RR and GLH can be expressed as $O(\log |\calX|)$ and $O(\log g)$, respectively. 
When $|\calX| < g$, the RR is better than the GLH in terms of the communication cost. 

\subsection{Implications for Privacy and Utility}
\label{sub:implications}
Theorem~\ref{thm:optimal_g} indicates that a larger $g$ results in a smaller $l_2$ loss when 
$(\calU,\alpha)$-PIE privacy 
is used as a privacy 
metric. 
On the other hand, the optimal value of $g$ that minimizes the $l_2$ loss is $g=e^\epsilon+1$ \cite{Wang_USENIX17} when $\epsilon$-LDP is used as a privacy 
metric. 
This is counter-intuitive because reducing $g$ from a larger value to $e^\epsilon+1$ decreases the $l_2$ loss. 
In other words, compressing the personal data $x$ with a smaller $g$ results in higher utility until $g=e^\epsilon+1$. 

One explanation for this counter-intuitive result is as follows. 
LDP is \textit{data privacy} that makes the original data $X$ indistinguishable from any other possible data in $\calX$. 
Here it becomes more difficult to guarantee the indistinguishability as the size $|\calX|$ of the data domain increases. 
Thus, for a fixed $\epsilon$, a larger $|\calX|$ results in \textit{lower utility}. 
In fact, by (\ref{eq:RR}), the RR with the same $\epsilon$ requires more noise for a larger $|\calX|$. 
A similar phenomenon occurs in matrix factorization under LDP \cite{Shin_TKDE18} -- a large amount of noise is added in \cite{Shin_TKDE18} due to high dimensionality of data. 
This is caused by the fact that LDP guarantees the indistinguishability of data. 
One way to increase the utility while fixing $\epsilon$ is \textit{dimension reduction}; i.e., reducing the value of $g$ (dimension reduction is also adopted in  \cite{Shin_TKDE18}). 
It should be noted, however, that too small $g$ results in a significant loss of information about the personal data $x$. 
The optimal $g$, which balances these two effects, is $g=e^\epsilon+1$.

Our PIE privacy has a different implication for privacy and utility. 
PIE privacy is \textit{user privacy} that aims to prevent the identification of $U$ (rather than the inference of $X$) by bounding $I(U;Y)$.  
Consequently, the MI loss parameter $\theta_{GLH}$ $(= \frac{e^\epsilon-1}{g+e^\epsilon-1})$ in the GLH 
does not depend on $|\calX|$. 
Therefore, for a fixed $\theta_{GLH}$, dimension reduction only results in the loss of information about $x$. 
This explains the reason 
for an intuitive result 
that a larger $g$ provides better utility in our notion. 

\section{Experimental Evaluation}
\label{sec:exp}

We 
proposed the PIE as a 
measure 
of re-identification risks in the local model, 
and 
analyzed the PIE and utility for the RR and GLH. 
Based on this, 
we would like to pose the following 
two basic 
questions:
1) How identifiable are personal data such as location traces and rating history? 
2) Is the PIE able to guarantee low re-identification risks for local obfuscation mechanisms such as the RR and GLH while keeping high utility? 
We conducted experiments to answer to these questions.

\subsection{Experimental Set-up}
\label{sub:set-up}
In our experiments, 
we used 
five 
large-scale\footnote{The biometric datasets are much smaller than the Foursquare and MovieLens datasets. This is because it is 
very hard to collect large-scale biometric data. 
We emphasize that 
each biometric dataset used 
in our experiments is one of the largest biometric datasets; e.g., much larger than 
\cite{FVC2006,Poh_PR10,XM2VTS,FERET,SDUMLA-HMT,Kumar_TIP12,Vanoni_BIOMS14}.} datasets:  

\smallskip
\noindent{\textbf{Location Trace.}}~~As a location trace dataset (denoted by \textbf{LT}), we used the Foursquare dataset (Global-scale Check-in Dataset with User Social Networks) \cite{Yang_WWW19}. 
This dataset 
includes 
$90048627$ check-ins by $2733324$ users 
on POIs all over the world. 
Each check-in is associated with its timestamp. 
We extracted $1370637$ users who had at least $10$ check-ins. 
The total number of POIs checked in by these users was $10500393$; i.e., $|\calU|=1370637$, $|\calX|=10500393$. 

For each user $u_i\in\calU$, we divided the location trace (time-series check-in data) into two disjoint traces of the same size. 
We used the former trace as training data in the partial-knowledge model, and the latter trace as personal data $X$. 
We used the training data of user $u_i$ as auxiliary information $r_i$ available to the adversary. 

\smallskip
\noindent{\textbf{Rating History.}}~~As a rating history dataset (denoted by \textbf{RH}), we used the MovieLens Latest dataset \cite{MovieLens}, which includes $27753444$ ratings by $283228$ users. 
Ratings are made on a 5-star scale with half-star increments ($0.5$ to $5$). 
Each rating is associated with its timestamp. 
We used $4$ to $5$ star ratings that represent ``likes'' and extracted $185107$ users who provided at least $10$ such ratings. 
The total number of movies rated by these users 
was $58098$; i.e., $|\calU|=185107$, $|\calX|=58098$. 

Since each user rates each movie at most once, it is difficult for the adversary to identify a user based on training data completely separated from personal data $X$. 
Therefore, for each user $u_i\in\calU$, we used the whole rating history (time-series rating data) as personal data $X$, 
and 
used the first $1$ to $5$ events (ratings) in $X$ as training data in the partial-knowledge model. 

\smallskip
\noindent{\textbf{Face.}}~~The NIST BSSR1 Set3 dataset \cite{BSSR} includes face scores (similarities) from $3000$ users; i.e., $|\calU|=3000$. 
Face scores were calculated by two matchers (``C'' and ``G''). 
We used scores from the matcher G because some users had inappropriate scores ($= -1$) in the matcher C. 
We used $3000 \times 3000$ scores in total. 
The face matcher G has lower errors than the best matcher in the FRPC 2017 \cite{Grother_NISTIR17}, as described in Section~\ref{sub:results_non-private} in detail.

\smallskip
\noindent{\textbf{Fingerprint.}}~~The CASIA-FingerprintV5 dataset \cite{CASIA}  (denoted by \textbf{FP}) includes $20000$ fingerprint images of $4000$ fingers ($5$ images per finger). 
We assumed that each finger was presented by a different user; i.e., $|\calU|=4000$. 
For each finger, we used the first and second images as a template and biometric data presented at authentication, respectively. 
To calculate a score between two fingerprint images, we used the VeriFinger SDK 7.0 \cite{VeriFinger}, a state-of-the-art commercial fingerprint matcher. 
We extracted $4000\times4000$ scores in total.

\smallskip
\noindent{\textbf{Finger-Vein.}}~~The finger-vein dataset in \cite{Yanagawa_BIC07} (denoted by \textbf{FV})
includes $33330$ finger-vein images of $3030$ fingers ($11$ images per finger); i.e., $|\calU|=3030$. 
For each finger, we used two images as a template and biometric data at authentication, respectively. 
To calculate a score, we used the CIRF (Correlation Invariant Random Filtering) \cite{biometric_security_ch1,Takahashi_IETBio12}. 
We used $3030\times3030$ scores between the transformed finger-vein features. 

\smallskip
To calculate scores in \textbf{LT} and \textbf{RH}, we 
used the matching algorithm based on the Markov chain model, 
because this model is effective for location privacy attacks 
\cite{Gambs_JCSS14,Mulder_WPES08,Murakami_TIFS17,Shokri_PETS11} and personalized item recommendation \cite{Rendle_WWW10,Cai_IJCAI17}. 
Specifically, we trained a transition matrix $\bmLamb_i\in[0,1]^{|\calX|\times|\calX|}$ for each user $u_i$ from training data via the MLE (Maximum Likelihood Estimation). 
We also trained a visit probability vector $\bmpi_i\in\calC$, which represents a probability distribution of personal data in $\calX$, for each user $u_i$ via the MLE; i.e., $\bmpi_i$ is an empirical distribution of the training data. 
The auxiliary information $r_i$ of the user $u_i$ can be expressed as: $r_i = (\bmLamb_i, \bmpi_i)$. 
Given obfuscated data $Y$ and $r_i$, 
we calculated a likelihood that $Y$ belongs to $u_i$ 
as follows. 
We calculated the likelihood for the first event in $Y$ via $\bmpi_i$ and for the subsequent events in $Y$ via $\bmLamb_i$. 
Then we multiplied them to obtain the likelihood for $Y$. 
Here we assigned a small positive value ($=10^{-8}$) to zero elements in $\bmLamb_i$ and $\bmpi_i$ so that the likelihood never becomes $0$ \cite{Mulder_WPES08,Murakami_TIFS17}. 
We used the likelihood for $Y$ as a score $s_i$. 

Based on genuine scores (scores for the same user) and impostor scores (scores for different users) for $|\calU|$ users, we evaluated the PSE 
by estimating 
$D(f_G||f_I)$. 
Note that $D(f_G||f_I)$ measures the PSE for a very large value of $n$ (when $n$ goes to infinity), as in Theorem~\ref{thm:BSE_KL}. 
To estimate $D(f_G||f_I)$, we used the generalized $k$-NN estimator \cite{Wang_TIT09} because it is asymptotically unbiased; i.e., it converges to the true value as the sample size increases. 
We also confirmed that for each dataset, $D(f_G||f_I)$ converges as the number of users increases. 

\subsection{Results for No Obfuscation Cases}
\label{sub:results_non-private}

We first evaluated the identifiability of personal data 
when no obfuscation mechanisms were used. 
Figure~\ref{fig:res1_PSEmax} shows the PSE in the maximum-knowledge model \cite{Domingo-Ferrer_PST15,Ruiz_PSD18} where 
the adversary uses personal data $X$ (the latter half of the trace in \textbf{LT} and the whole rating history in \textbf{RH}) as training data. 
In the left figure, each user sends 
$X$ without obfuscation; i.e., $Y=X$. 
In the right figure, each user sends the first $1$ to $5$ events in $X$ as $Y$. 
We note again that we measure the PSE for a very large value of $n$ (when $n$ goes to infinity); e.g., in Figure~\ref{fig:res1_PSEmax}, the PSE of \textbf{FV} is larger than $\log 3030 = 11.6$.

\begin{figure}
\centering
\includegraphics[width=0.75\linewidth]{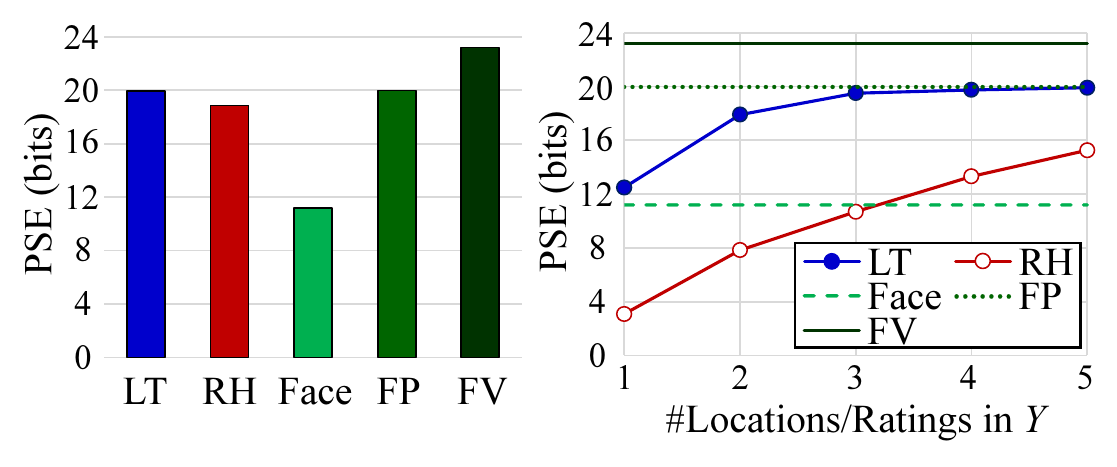}
\vspace{-2mm}
\caption{PSE in the maximum-knowledge attacker model. 
In \textbf{LT} and \textbf{RH}, the adversary uses $X$ as training data. 
Then in the left figure, each user sends $X$ without obfuscation ($Y=X$). In the right figure, each user sends the first $1$ to $5$ events in $X$ as $Y$.}
\label{fig:res1_PSEmax}
\end{figure}

The left figure shows that 
\textbf{LT} 
and 
\textbf{RH} 
have almost the same identifiability as the commercial fingerprint matcher (\textbf{FP}) in the maximum-knowledge model. 
Moreover, the right figure shows that only three locations are enough to have almost the same identifiability as a fingerprint. 
This is consistent with the fact that only three locations are enough to uniquely characterize about $80\%$ of individuals among one and a half million people \cite{Montjoye_SR13}. 

Although the maximum-knowledge model reflects the \textit{worst-case scenario} where the maximum auxiliary information is available to the adversary, 
it may not be realistic. 
For example, it is natural to consider that the whole location trace $X$ is not available as auxiliary information in practice 
(unless the adversary tracks the user all the time). 
In this case, even if $X$ is highly unique (as shown in \cite{Montjoye_SR13} for location data), it might be difficult for the adversary to identify a user. 

Therefore, we evaluated the identifiability of \textbf{LT} and \textbf{RH} in the partial-knowledge model. 
In \textbf{LT}, we extracted $1896$ users who had at least $1000$ check-ins. 
Then for each user, we used the former half of the trace as training data, and the latter half of the trace as personal data $X$. 
We changed the number of training events (locations) from $100$ to $500$. 
In \textbf{RH}, we extracted $1888$ users who had at least $500$ ratings with $4$ to $5$ stars. 
Then we used the first $1$ to $5$ events (ratings) in $X$ as training data. 
In both \textbf{LT} and \textbf{RH}, we used no obfuscation mechanisms ($Y=X$).

\begin{figure}[t]
\centering
\includegraphics[width=0.75\linewidth]{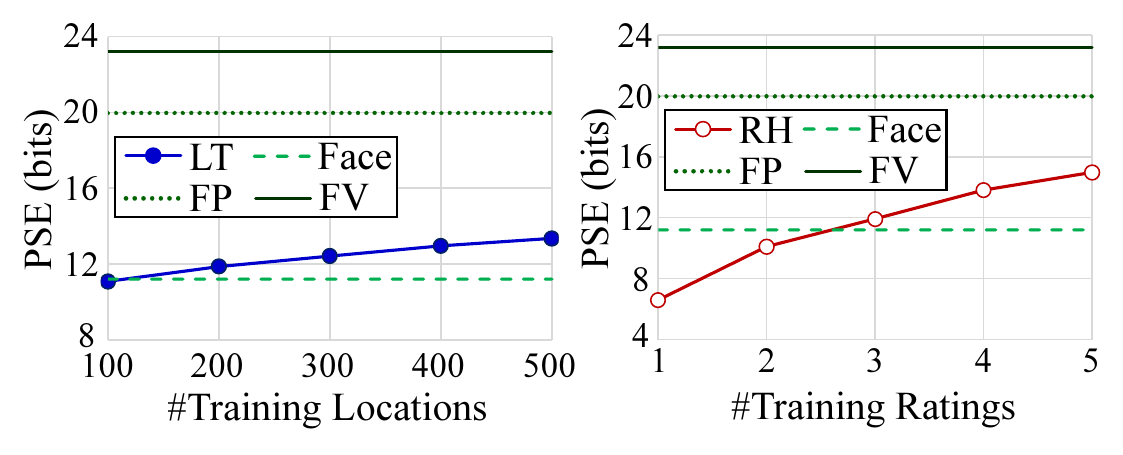}
\vspace{-2mm}
\caption{PSE in the partial-knowledge attacker model. 
In \textbf{LT}, the adversary uses training data separated from $X$. 
In \textbf{RH}, the adversary uses the first $1$ to $5$ ratings in $X$ as training data. 
Then each user sends $X$ without no obfuscation; i.e., $Y=X$.}
\label{fig:res2_PSEpart}
\end{figure}
\begin{figure}[t]
\centering
\includegraphics[width=0.55\linewidth]{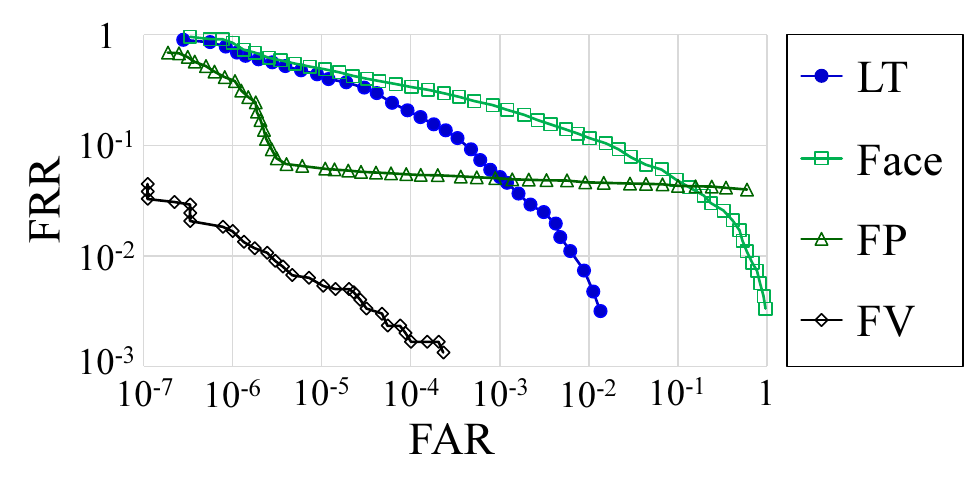}
\vspace{-2mm}
\caption{DET curve in the partial-knowledge attacker model. In \textbf{LT}, the number of training locations is $500$.}
\label{fig:res3_DET}
\end{figure}

Figure~\ref{fig:res2_PSEpart} shows the results. In both \textbf{LT} and \textbf{RH}, the PSE increases with increase in the number of training events. 
The PSE of \textbf{LT} (resp.~\textbf{RH}) is larger than that of the face matcher in  \cite{BSSR} when the number of training events is larger than or equal to $200$ (resp.~$3$). 
We emphasize that although a part of $X$ was used as training data in \textbf{RH}, the training data in \textbf{LT} was completely separated from $X$. 

We also evaluated FAR (False Acceptance Rate) and FRR (False Rejection Rate), 
commonly used accuracy measures in biometrics \cite{guide,intro}. FAR (resp.~FRR) is the proportion of verification attempts in which the system incorrectly accepts an impostor (resp.~rejects a genuine user). 
In the local privacy model, FAR (resp.~FRR) corresponds to the error rates in which, given auxiliary information $r_i$ of some user, the adversary incorrectly decides that obfuscated data $Y$ and $r_i$ belong to the same user (resp.~different users). 
We evaluated the DET (Detection Error Tradeoff) curve \cite{Mansfield_NPL02}, which is obtained by plotting FAR against FRR at various thresholds, using genuine and impostor scores for $|\calU|$ users. 

Figure~\ref{fig:res3_DET} shows the DET curve of \textbf{LT} in the partial-knowledge model ($500$ training locations) and biometric data. 
This figure shows that \textbf{LT} provides smaller FAR and FRR than the face, 
which is consistent with the PSE in Figure~\ref{fig:res2_PSEpart}. 

We also note that 
the best matcher in the FRPC (Face Recognition Prize Challenge) 2017 \cite{Grother_NISTIR17} had FRR of about $0.15$ (resp.~$0.1$) at FAR of $0.01$ (resp.~$0.1$), which is \textit{worse} than that of the face matcher in Figure~\ref{fig:res3_DET}. 
The high FAR and FRR in the FRPC 2017 were caused by the fact that face images were collected without tight quality constraints. 
In particular, unconstrained yaw and pitch pose variation caused errors \cite{Grother_NISTIR17}. 
Similarly, Figure~\ref{fig:res3_DET} shows that FRR of 
\textbf{FP} 
is high (even using the commercial fingerprint matcher). 
This is because 
users were asked to rotate their fingers with various levels of pressure in 
the CASIA-FingerprintV5 dataset \cite{CASIA}. 

In summary, our answers to the first question at the beginning of Section~\ref{sec:exp} 
are 
as follows: 
\begin{itemize}
\item Three locations are enough to have almost the same identifiability as the commercial fingerprint matcher in the maximum-knowledge model (ratings need more events).
\item A long location trace ($\geq 500$ locations) has higher identifiability than the face matcher in \cite{NIST04} and the best face matcher in the FRPC 2017 \cite{Grother_NISTIR17} in the partial-knowledge model where training data is separated from $X$. 
\end{itemize}

We emphasize that the second answer illustrates an interesting feature of our PSE that provides a new intuitive understanding of re-identification risks; i.e., \textit{a long location trace is more identifiable than the best face matcher in the prize challenge}. 
For example, the EU's AI Act \cite{EU_AI_Act} states that (both `real-time' and `post') remote biometric identification systems should be classified as high-risk. Based on the second answer, we argue that systems collecting location traces should also be considered as high-risk in terms of the re-identification risk.

\subsection{Results for the RR and GLH}
\label{sub:results_RR-GLH}
Next we evaluated the privacy and utility of the RR and GLH. 
Here we focused on \textbf{LT} because it has higher identifiability than \textbf{RH} as shown in Figure~\ref{fig:res1_PSEmax}. 

We used location traces of all users ($|\calU|=1370637$). 
We assumed that each user obfuscates the first location $X$ in the latter half of the trace using the RR or GLH, and sends the obfuscated location $Y$ to a data collector. 
In other words, we assumed that each user sends a single datum as in Section~\ref{sub:util_anal}. 
Note that a single location may be enough to identify a user, as shown in Figure~\ref{fig:res1_PSEmax}. 
For example, a user's home or work location is highly identifiable information \cite{Golle_Pervasive09}. 
We also note that the privacy of the RR and GLH in the case of multiple obfuscated data can be discussed based on the composition theorems (Theorems~\ref{thm:Comp_RR} and \ref{thm:Comp_GLH}). 

We first 
examined 
how tight our upper-bounds on the PIE for the RR and GLH (Theorems~\ref{thm:PIE_RR} and \ref{thm:PIE_GLH}) are. 
To this end, we evaluated the PSE, which lower-bounds the PIE 
(Proposition~\ref{prop:PIE_PSE}). 
As training data of the adversary, we used the former half of the trace; i.e., partial-knowledge model. 

Figure~\ref{fig:res4_RR_GLH_PIE} shows the upper-bounds ($\alpha$ in Proposition~\ref{prop:LDP_PIE}, Theorem~\ref{thm:PIE_RR}, and Theorem~\ref{thm:PIE_GLH}) and lower-bounds (PSE) on the PIE. 
Here we set $g$ in the GLH to $g=10^8$. 
Figure~\ref{fig:res4_RR_GLH_PIE} shows that the values of $\alpha$ in Theorems~\ref{thm:PIE_RR} and \ref{thm:PIE_GLH} are much 
smaller 
than $\alpha$ in Proposition~\ref{prop:LDP_PIE}, and 
are close to the PSE. 
This indicates that our upper-bounds for the RR and GLH in Theorems~\ref{thm:PIE_RR} and \ref{thm:PIE_GLH} are fairly tight and cannot be improved much. 

We also 
examined 
how tight the lower-bounds on the identification error probability $\beta_{U|\bmS}$ in Proposition~\ref{prop:identification_error} are. 
To this end, we performed the re-identification attack for each obfuscated location $Y$ using the matching algorithm based on the Markov chain model and the best score rule (described in Section~\ref{sub:framework}). 
Then we evaluated the identification error rate, which is the proportion of correct identification results. 

\begin{figure}[t]
\centering
\includegraphics[width=0.75\linewidth]{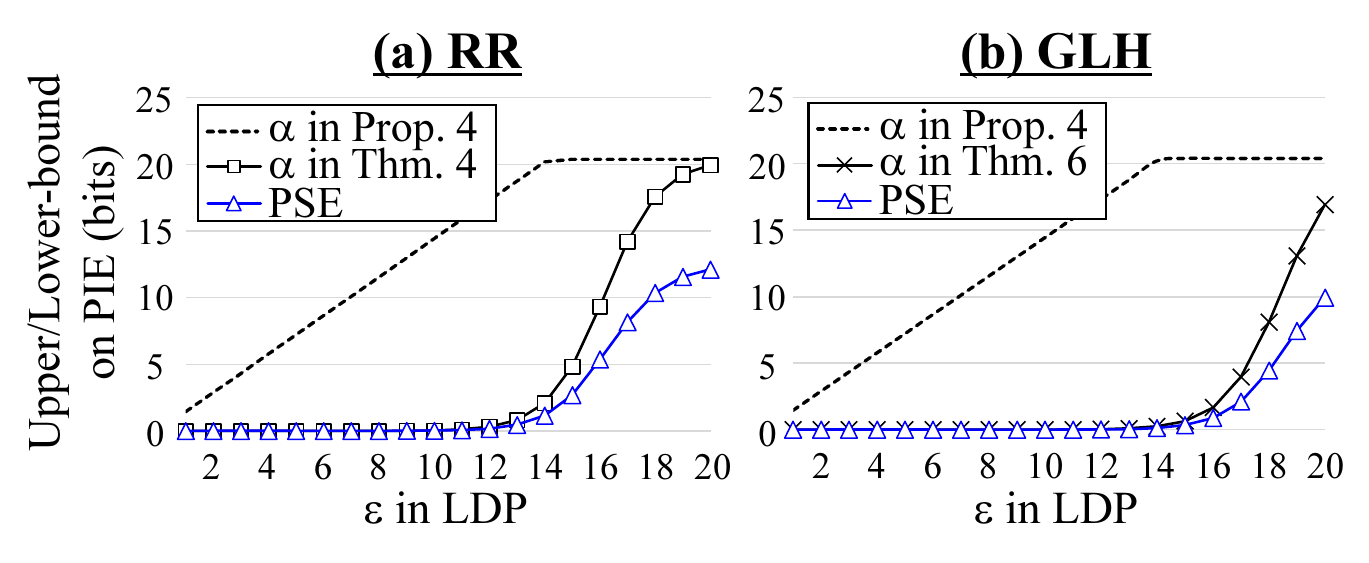}
\vspace{-2mm}
\caption{Upper-bounds ($\alpha$ in Proposition~\ref{prop:LDP_PIE} and Theorems~\ref{thm:PIE_RR} and \ref{thm:PIE_GLH}) and lower-bounds (PSE) on the PIE in the $\epsilon$-RR and $(g,\epsilon)$-GLH ($g=10^8$).}
\label{fig:res4_RR_GLH_PIE}
\vspace{2mm}
\includegraphics[width=0.75\linewidth]{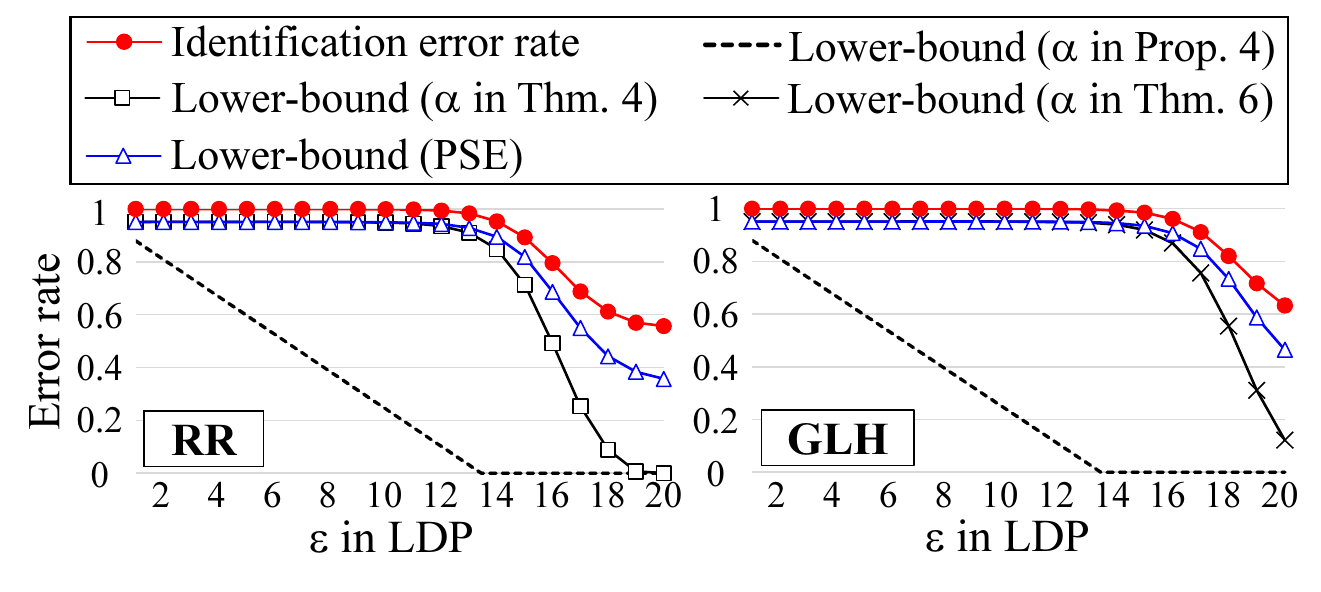}
\vspace{-2mm}
\caption{Identification error rate in the $\epsilon$-RR and $(g,\epsilon)$-GLH ($g=10^8$). 
The lower-bounds (blue and black lines) are obtained by assigning 
$\alpha$ (resp.~PSE) in Figure~\ref{fig:res4_RR_GLH_PIE} to (\ref{eq:Fano_uniform_cor}) (resp.~(\ref{eq:Fano_uniform})). }
\label{fig:res4_RR_GLH_IDerror}
\end{figure}

Figure~\ref{fig:res4_RR_GLH_IDerror} shows the results. 
Here the lower-bounds 
on $\beta_{U|\bmS}$ are 
obtained by 
assigning 
$\alpha$ (resp.~PSE) in Figure~\ref{fig:res4_RR_GLH_PIE} to (\ref{eq:Fano_uniform_cor}) (resp.~(\ref{eq:Fano_uniform})).
Figure~\ref{fig:res4_RR_GLH_IDerror} shows that the gap between the identification error rate and the lower-bound is caused by two factors. 
Specifically, the gap between ``Identification error rate'' (red line) and ``Lower-bound (PSE)'' (blue line) is caused by the generalized Fano's inequality \cite{Han_TIT94}; i.e., the first inequality in (\ref{eq:Fano_uniform}). 
The gaps between ``Lower-bound (PSE)'' (blue line) and the other lower-bounds (black lines) are caused by the gap between the PSE and 
$\alpha$ in PIE privacy 
($\alpha$ in Proposition~\ref{prop:LDP_PIE} and Theorems~\ref{thm:PIE_RR} and \ref{thm:PIE_GLH}).
Figure~\ref{fig:res4_RR_GLH_IDerror} also shows that 
the lower-bounds by Theorem~\ref{thm:PIE_RR} and \ref{thm:PIE_GLH} are much tighter than the lower-bound by Proposition~\ref{prop:LDP_PIE}. 

Finally, we evaluated the utility. 
As a task for a data collector, we considered distribution estimation. 
For a distribution estimation method, we used 
the empirical estimator without a significant threshold (denoted by \textbf{emp}) and the empirical estimator with a significant threshold \cite{Erlingsson_CCS14} (denoted by \textbf{emp+thr}). 
In \textbf{emp+thr}, we set the significance level to $0.05$ in the same way as \cite{Erlingsson_CCS14,Murakami_USENIX19,Wang_USENIX17}, and assigned zero to probabilities below a significance threshold. 
Then we evaluated the sum of 
$l_2$ losses $(\bmp(x)-\hbmp(x))^2$ over the top $\phi\in\{100,|\calX|\}$ POIs whose probabilities $\bmp(x)$ are the largest. 
The top $\phi=100$ POIs correspond to heavy hitters. 

\begin{figure}[t]
\centering
\includegraphics[width=0.8\linewidth]{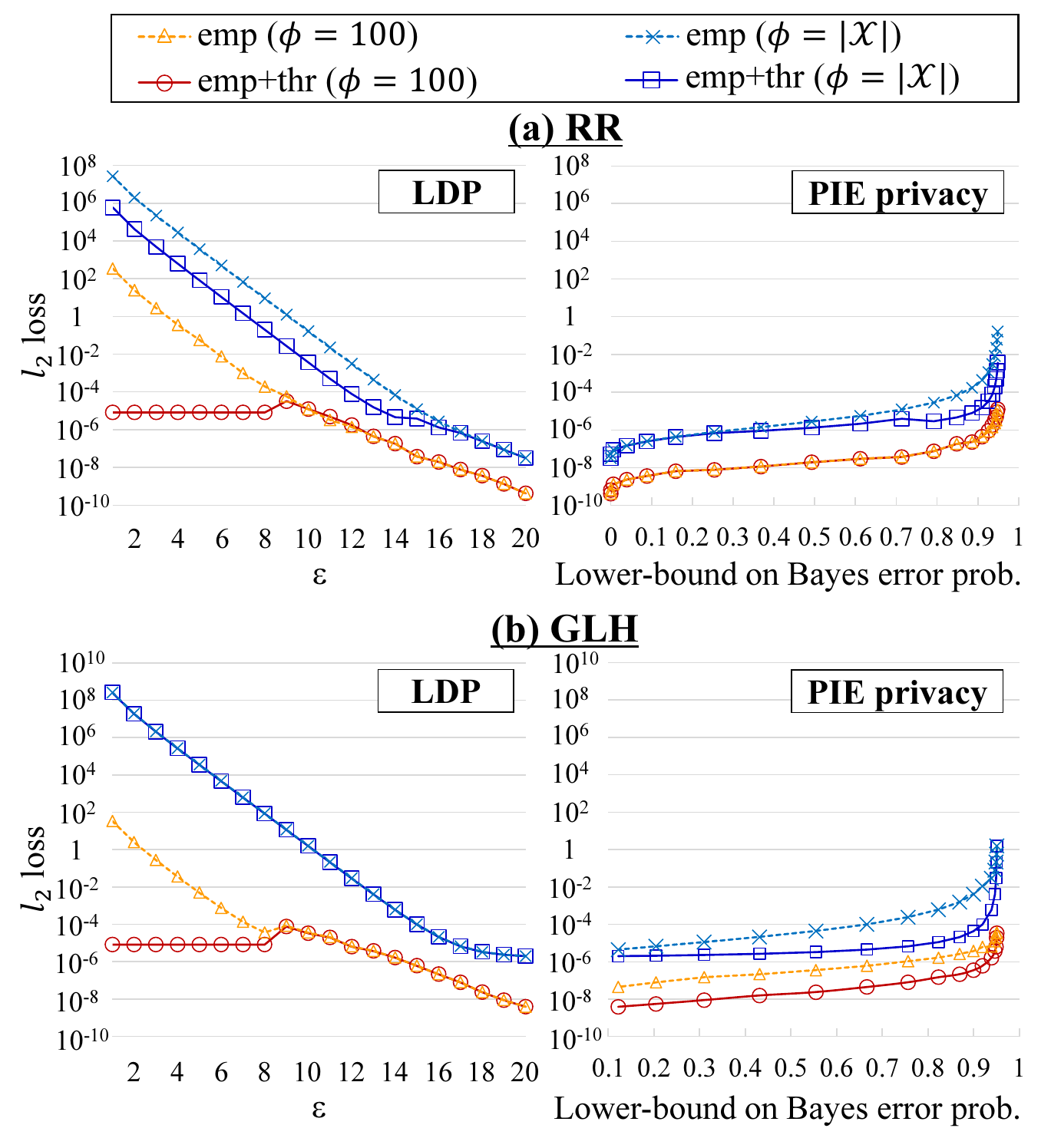}
\vspace{-2mm}
\caption{Sum of the 
$l_2$ losses over the top $\phi\in\{100,|\calX|\}$ POIs with various values of $\epsilon$ ($g=10^8$).}
\label{fig:res5_DistErr_eps}
\end{figure}
\begin{figure}[t]
\centering
\includegraphics[width=0.7\linewidth]{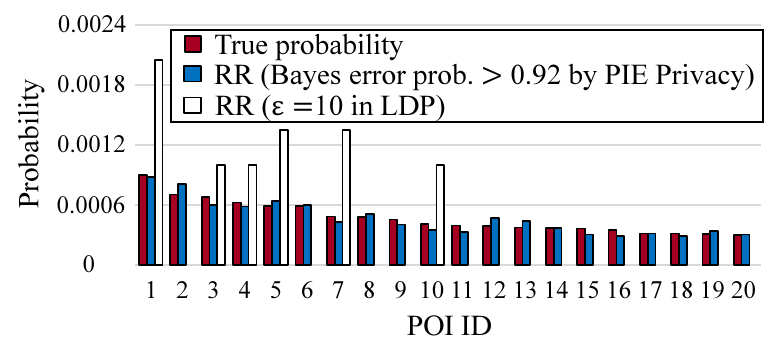}
\vspace{-4mm}
\caption{Probabilities of the top-$20$ POIs and estimates 
by \textbf{emp+thr}. 
The relative error of the RR with 
$\beta_{U|\bmS} > 0.92$ guaranteed by PIE privacy 
(resp.~$\epsilon=10$ in LDP) was 
$0.10$ 
(resp.~$1.05$)}.

\label{fig:res6_DistVisual}
\end{figure}

Figure~\ref{fig:res5_DistErr_eps} shows the results ($g=10^8$). 
Here the left figures show $\epsilon$ in LDP, whereas 
the right figures show the lower-bounds on the Bayes error probability $\beta_{U|\bmS}$ obtained by PIE privacy (Corollary~\ref{cor:identification_error_cor} 
and Theorems~\ref{thm:PIE_RR} and \ref{thm:PIE_GLH}). 
Figure~\ref{fig:res5_DistErr_eps} shows that the utility is improved by using the significant threshold. 
This is because 
\textbf{emp} assigns negative values to many 
input symbols 
\cite{Agrawal_SIGMOD05}. 
However, 
\textbf{emp+thr} still provides poor utility when $\epsilon \leq 10$. 
We also show in Figure~\ref{fig:res6_DistVisual} the probabilities $\bmp(x)$ of the top-$20$ POIs (red bars)
and the estimates $\hbmp(x)$ by 
\textbf{emp+thr} when we use the RR with $\epsilon=10$ (white bars). 
The RR with $\epsilon=10$ performs poorly. 
The relative error 
$\frac{| \bmp(x)-\hbmp(x) |}{\bmp(x)}$ 
of the RR with $\epsilon=10$ for 
each of 
the top $20$ POIs 
was on average $1.05$ ($>1$). 
Note that $\epsilon=10$ is considered to provide almost no privacy guarantees for personal data \cite{DP_Li}, because $e^\epsilon$ in (\ref{eq:LDP}) is $e^{10} = 22026$. 
This means that LDP fails to guarantee meaningful privacy and utility for this task. 

On the other hand, our PIE privacy can be used to prevent re-identification attacks while keeping high utility. 
For example, Figure~\ref{fig:res5_DistErr_eps}(a) shows that LDP requires $\epsilon > 12$ (no meaningful privacy guarantees) to achieve the $l_2$ loss of $10^{-6}$ using \textbf{emp+thr} ($\phi = 100$). 
In contrast, PIE privacy guarantees the Bayes error probability to be $\beta_{U|\bmS} > 0.92$ with the same $l_2$ loss. 
Figure~\ref{fig:res6_DistVisual} shows the estimates $\hbmp(x)$ by the RR and \textbf{emp+thr} 
in this case 
(blue bars). 
The RR accurately estimates the distribution for the $20$-POIs while satisfying $\beta_{U|\bmS} > 0.92$. 
The relative error of the RR in this case 
was 
$0.10$, 
which was much smaller than $1$. 
Therefore, PIE privacy 
can be used to guarantee the identification error probability larger than 
$0.92$ 
while keeping high utility in this task.

We finally set the MI loss parameters in the RR and GLH to $\theta_{RR}=\theta_{GLH}=0.5$ (which guarantees 
the Bayes error probability $\beta_{U|\bmS}$ larger than $0.5$), 
and changed the parameter $g$ in the GLH from $10^3$ to $2 \times 10^9$ while fixing the MI loss parameter $\theta_{GLH}$. 
Then we evaluated the sum of 
the 
$l_2$ losses over the top $\phi\in\{100,|\calX|\}$ POIs. 

\begin{figure}[t]
\centering
\includegraphics[width=0.75\linewidth]{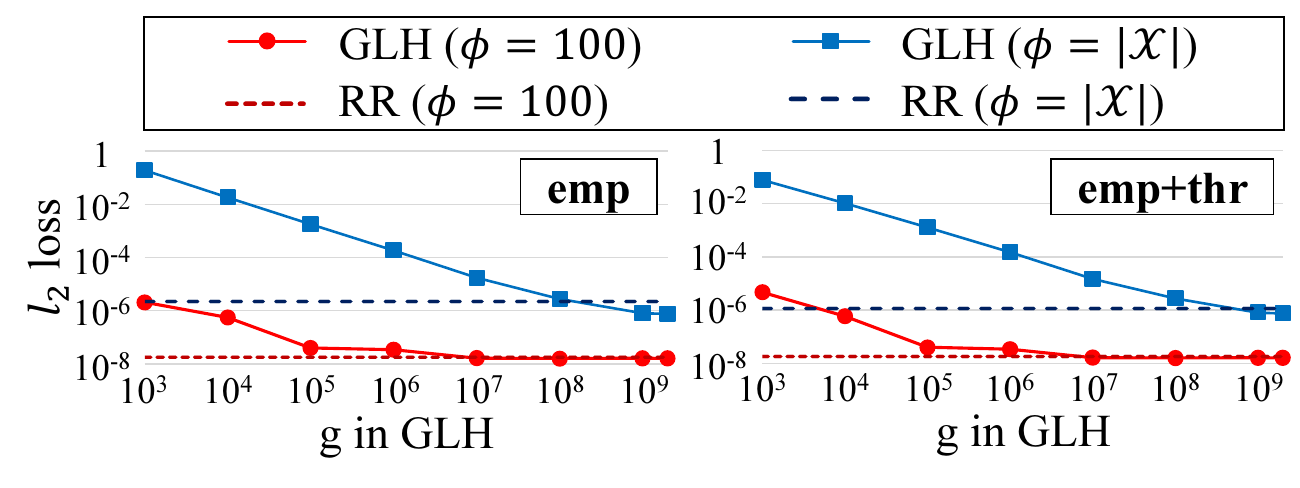}
\vspace{-2mm}
\caption{Sum of the 
$l_2$ losses over the top $\phi\in\{100,|\calX|\}$ POIs with various values of $g$ when the Bayes error probability $\beta_{U|\bmS}$ larger than $0.5$ ($\theta_{RR}=\theta_{GLH}=0.5$).}
\label{fig:res5_DistErr_g}
\end{figure}

Figure~\ref{fig:res5_DistErr_g} shows the result. 
In the GLH, 
the $l_2$ loss decreases 
with increase in $g$, which is consistent with Theorem~\ref{thm:optimal_g}. 
The GLH with a large $g$ has almost the same utility as the RR for $\phi=100$. 
When we use \textbf{emp}, the GLH with a large $g$ has slightly better utility than the RR for $\phi=|\calX|$. 
However, when we use \textbf{emp+thr}, they have almost the same utility for $\phi=|\calX|$ because zero values are assigned to unpopular 
symbols 
in both the RR and GLH, as discussed in Section~\ref{sub:util_anal}. 

In summary, our answers to the second 
question 
at the beginning of Section~\ref{sec:exp} 
are as follows: 
\begin{itemize}
\item Our PIE guarantees the identification error probability larger than 
$0.92$ 
while keeping high utility in distribution estimation. 
On the other hand, 
LDP 
destroys the utility even when $\epsilon = 10$.
\item 
Compressing the personal data $X$ with a smaller 
value of 
$g$ in the GLH results in the loss of utility 
for given PIE guarantees, which is consistent with our theoretical results. 
\end{itemize}

\section{Conclusion}
\label{sec:conclusion}
We 
proposed the PIE (Personal Information Entropy) as a 
measure 
of re-identification risks in the local model. 
We conducted experiments using 
five 
datasets, and showed that a location trace has higher identifiability than the face matcher in \cite{BSSR} and the best matcher in the FRPC 2017 \cite{Grother_NISTIR17} in the partial-knowledge model. 
We also showed that 
the PIE can be used to guarantee low re-identification risks for the RR and GLH 
while keeping high utility in distribution estimation. 

As described in Section~\ref{sec:intro}, our PIE privacy is an average notion 
(though it differs from the re-identification rate in that PIE privacy does not specify an identification algorithm nor the adversary's background knowledge, as described in Section~\ref{sub:PIE}). 
One way to extend our PIE privacy from the average notion to the worse-case notion is to use \textit{$\alpha$-mutual information} \cite{Verdu_ITA15}, which is an extension of the mutual information using R\'{e}nyi divergence. 
R\'{e}nyi divergence with $\alpha=1$ is equivalent to the KL divergence, whereas R\'{e}nyi divergence with $\alpha=\infty$ is equivalent to the max divergence, which is a worst-case analog of the KL divergence. 
Therefore, we can extend 
our PIE privacy 
to the the worst-case privacy 
metric 
by using $\alpha$-mutual information with a large $\alpha$. 
Then the following questions remain open: 
How do the theoretical properties of the PIE (shown in Section~\ref{sub:properties}) change? 
How do the bounds on privacy and utility for 
obfuscation 
mechanisms (e.g., RR, GLH) change? 
As future work, we would like to 
investigate these open questions. 

\section*{Acknowledgment}
The authors would like to thank Casey Meehan (UCSD) for technical comments on this paper.

\bibliographystyle{plain}
\bibliography{main}

\appendix

\arxiv{\section{Proofs of Statements in Section 3}}
\conference{\section{Proofs of Statements in Section~\ref{sec:unified}}}
\label{sec:proof_sec4}

\begin{proof}[Proof of Theorem~\ref{thm:post-processing} (Post-processing invariance)]
Let $\calY'$ be the range of the randomized algorithm $\lambda$, and $Y'$ be a random variable representing obfuscated data output by $\lambda \circ \bmQ$. 

The obfuscation mechanism $\bmQ$ takes as input personal data $X$ of a user $U$ and outputs $Y$. 
Then the randomized algorithm $\lambda$ takes as input $Y$ and outputs $Y'$. 
Thus, $U$, $X$, $Y$ and $Y'$ form the Markov chain; i.e., $U \rightarrow X \rightarrow Y \rightarrow Y'$. 
Then by the data processing inequality \cite{elements} and (\ref{eq:PIE-privacy}), we obtain:
\begin{align*}
I(U;Y') \leq I(U;Y) \leq \alpha ~~ (bits).
\end{align*}
Therefore, $\lambda \circ \bmQ$ provides 
$(n,\alpha)$-PIE privacy.
\end{proof}

\begin{proof}[Proof of Theorem~\ref{thm:convexity} (Convexity)]
Let $Y_w$, $Y_1$, and $Y_2$ be random variables representing obfuscated data output by $\bmQ_w$, $\bmQ_1$, and $\bmQ_2$, respectively. 
Given a user $U$, let $p_{X|U}$, $p_{Y_w|U}$, $p_{Y_1|U}$, $p_{Y_2|U}$ be a probability distribution of $X$, $Y_w$, $Y_1$, and $Y_2$, respectively. 
Then the probability of outputting $y$ from $\bmQ_w$ is written as follows:
\begin{align}
p_{Y_w|U}(y) 
&= \sum_{x\in\calX} \bmQ_w(y|x) p_{X|U}(x) \nonumber\\
&= \sum_{x\in\calX} \left(w \bmQ_1(y|x) + (1-w) \bmQ_2(y|x)\right) p_{X|U}(x) \nonumber\\
&= w p_{Y_1|U}(y) + (1-w) p_{Y_2|U}(y). 
\label{eq:p_YP_U}
\end{align}
The mutual information $I(U;Y_w)$ is convex in $p_{Y_w|U}$ (see Theorem~2.7.4 in \cite{elements}). 
In addition, $I(U; Y_1) \leq \alpha$ and $I(U; Y_2) \leq \alpha$ because $\bmQ_1$ and $\bmQ_2$ provide 
$(n,\alpha)$-PIE privacy. 
Then by (\ref{eq:p_YP_U}), we obtain:
\begin{align*}
I(U;Y_w) 
&\leq w I(U;Y_1) + (1-w) I(U;Y_2) \\
&\leq w \alpha + (1-w) \alpha \\
& = \alpha.
\end{align*}
Therefore, $\bmQ_w$ provide 
$(n,\alpha)$-PIE privacy. 
\end{proof}

\arxiv{\section{Proofs of Statements in Section 4}}
\conference{\section{Proofs of Statements in Section~\ref{sec:theoretical}}}
\label{sec:proof_sec5}
\begin{proof}[Proof of Proposition~\ref{prop:LDP_PIE} (LDP and PIE)]
We use the following lemma:
\begin{lemma}
\label{prop:LDP_MI} 
If an obfuscation mechanism $\bmQ$ provides $\epsilon$-LDP, then 
\begin{align}
I(X;Y) \leq \min\{\epsilon \log e,\epsilon^2 \log e\}~~ (bits).
\label{eq:LDP_MI}
\end{align}
\end{lemma}
This is a special case of Lemma~1 in \cite{Cuff_CCS16} in the local model. 
By (\ref{eq:PIE_obf}) in Proposition~\ref{prop:PIE_obf}, (\ref{eq:LDP_MI}), and $I(U;X) \leq \min\{\log n, \log |\calX|\}$, we obtain:
\begin{align*}
    I(U;Y) \leq \min\{\epsilon \log e, \epsilon^2 \log e, \log n, \log |\calX|\}.
\end{align*}
Note that this holds for any distribution 
$p_{U,X}$. 
Thus 
Proposition~\ref{prop:LDP_PIE} holds. 
\end{proof}

\begin{proof}[Proof of Theorem~\ref{thm:PIE_RR} (PIE of the RR)]
Let $p_U$ be a probability distribution of $U$. 
Let $p_Y$ be a probability distribution of $Y$, and $p_{Y|U=u_i}$ be a conditional probability distribution of $Y$ given $U=u_i$. 
Then $I(U;Y)$ can be written as follows:
\begin{align}
I(U;Y) = \sum_{i=1}^n p_U(u_i) D(p_{Y|U=u_i}||p_Y).
\label{eq:PIE_RR_IUY}
\end{align}
Let $p_{\text{uni}}$ be a uniform distribution over $\calX$; i.e., $p_{\text{uni}}(x) = \frac{1}{|\calX|}$ for any $x\in\calX$.  
Let $p_X$ be a probability distribution of $X$, and $p_{X|U=u_i}$ be a conditional probability distribution of $X$ given $U=u_i$. 
By (\ref{eq:RR}), we can regard the $\epsilon$-RR as a mechanism that given $x\in\calX$, outputs $y=x$ with probability $\theta_{RR}$ ($= \frac{e^\epsilon-1}{|\calX|+e^\epsilon-1}$) and outputs a value randomly chosen from $\calX$ (including $x$) with probability $1-\theta_{RR}$. 
Therefore, given $U=u_i$, $Y$ is generated from $p_{X|U=u_i}$ with probability $\theta_{RR}$, and is generated from $p_{\text{uni}}$ with probability $1-\theta_{RR}$. 
Then we obtain:
\begin{align}
&D(p_{Y|U=u_i}||p_Y) \nonumber\\
&= D(\theta_{RR} p_{X|U=u_i} \hspace{-1mm}+\hspace{-1mm} (1 \hspace{-1mm} - \hspace{-1mm} \theta_{RR}) p_{\text{uni}} || \theta_{RR} p_X \hspace{-1mm}+\hspace{-1mm} (1 \hspace{-1mm} - \hspace{-1mm} \theta_{RR}) p_{\text{uni}}) \nonumber\\
&\leq \theta_{RR} D(p_{X|U=u_i} || p_X) + (1 -\theta_{RR}) D(p_{\text{uni}} || p_{\text{uni}}) \nonumber\\
& \hspace{5mm} 
\text{(by the convexity of the KL divergence \cite{elements})} \nonumber\\
&= \theta_{RR} D(p_{X|U=u_i} || p_X). 
\label{eq:PIE_RR_D_pYuiY}
\end{align}
By (\ref{eq:PIE_RR_IUY}) and (\ref{eq:PIE_RR_D_pYuiY}), we obtain:
\begin{align}
I(U;Y) 
&\leq \sum_{i=1}^n p_U(u_i) \theta_{RR} D(p_{X|U=u_i} || p_X) \nonumber\\
&= \theta_{RR} \sum_{i=1}^n p_U(u_i) D(p_{X|U=u_i} || p_X) \nonumber\\
&= \theta_{RR} I(U;X).
\label{eq:PIE_RR_IUY_2}
\end{align}
Therefore, (\ref{eq:RR_IUY_IUX}) holds. 
Note that this holds for any distribution 
$p_{U,X}$. 
(\ref{eq:PIE_RR}) is immediately derived from (\ref{eq:RR_IUY_IUX}) because $I(U;X) \leq H(U) \leq \log n$ and $I(U;X) \leq H(X) \leq \log |\calX|$.
\end{proof}

\begin{proof}[Proof of Theorem~\ref{thm:Comp_RR} (Composition of the RR)]
For $i\in[t]$, let $\bmY^i = (Y^{(1)},\cdots,Y^{(i)})$. 
Then, by the chain rule for the mutual information \cite{elements}, we obtain:
\begin{align}
I(U;Y) = I(U;Y^{(1)}) + \sum_{i=2}^t I(U;Y^{(i)}|\bmY^{i-1}).
\label{eq:Comp_RR_IUY}
\end{align}
Let $p_U$ and $p_{\bmY^i}$ be a probability distribution of $U$ and $\bmY^i$, respectively. 
Let $p_{Y^{(i)}|U=u_i, \bmY^{i-1}=\bmy^{i-1}}$ (resp.~$p_{Y^{(i)}|\bmY^{i-1}=\bmy^{i-1}}$) be a conditional probability distribution of $Y^{(i)}$ given $U=u_i$ and $\bmY^{i-1}=\bmy^{i-1}$ (resp.~given $\bmY^{i-1}=\bmy^{i-1}$). 
Then $I(U;Y^{(i)}|\bmY^{i-1})$ in (\ref{eq:Comp_RR_IUY}) can be written as follows:
\begin{align}
&I(U;Y^{(i)}|\bmY^{i-1}) \nonumber\\
&= \sum_{\bmy^{i-1}\in\calY^{i-1}} p_{\bmY^{i-1}}(\bmy^{i-1}) \sum_{i=1}^n p_U(u_i|\bmy^{i-1}) 
D(p_{Y^{(i)}|U=u_i,\bmY^{i-1}=\bmy^{i-1}}||p_{Y^{(i)}|\bmY^{i-1}=\bmy^{i-1}}).
\label{eq:Comp_RR_IUYiYi-1}
\end{align}
Let $p_{\text{uni}}$ be a uniform distribution over $\calX$. 
Let $p_{X^{(i)}|U=u_i, \bmY^{i-1}=\bmy^{i-1}}$ (resp.~$p_{X^{(i)}|\bmY^{i-1}=\bmy^{i-1}}$) be a conditional probability distribution of $X^{(i)}$ given $U=u_i$ and $\bmY^{i-1}=\bmy^{i-1}$ (resp.~given $\bmY^{i-1}=\bmy^{i-1}$). 
By (\ref{eq:RR}), the $\epsilon$-RR takes as input $x\in\calX$, and outputs $y=x$ with probability $\theta_{RR}$ ($= \frac{e^\epsilon-1}{|\calX|+e^\epsilon-1}$) and outputs a value randomly chosen from $\calX$ (including $x$) with probability $1-\theta_{RR}$. 
Note that this is independent of the other input and output data. 
In other words, 
given $X^{(i)}$, the $\epsilon$-RR outputs $Y^{(1)}=X^{(i)}$ with probability $\theta_{RR}$ and a random value from $\calX$ (including $X^{(i)}$) with probability $1-\theta_{RR}$, 
\textit{irrespective of whether there are correlations between $X^{(1)}, \cdots, X^{(t)}$}. 
Thus we obtain:
\begin{align}
&D(p_{Y^{(i)}|U=u_i,\bmY^{i-1}=\bmy^{i-1}}||p_{Y^{(i)}|\bmY^{i-1}=\bmy^{i-1}}) \nonumber\\
&= D(\theta_{RR}~ p_{X^{(i)}|U=u_i, \bmY^{i-1}=\bmy^{i-1}} + (1-\theta_{RR}) p_{\text{uni}} 
|| \theta_{RR}~ p_{X^{(i)}|\bmY^{i-1}=\bmy^{i-1}} + (1-\theta_{RR}) p_{\text{uni}}) \nonumber\\
&\leq \theta_{RR} D(p_{X^{(i)}|U=u_i, \bmY^{i-1}=\bmy^{i-1}} || p_{X^{(i)}|\bmY^{i-1}=\bmy^{i-1}}) \nonumber\\
& \hspace{5mm} \text{(by the convexity of the KL divergence \cite{elements} and 
$D(p_{\text{uni}}||p_{\text{uni}})=0$)}. 
\label{eq::Comp_RR_DpYi}
\end{align}
By (\ref{eq:Comp_RR_IUYiYi-1}) and (\ref{eq::Comp_RR_DpYi}), we obtain:
\begin{align}
&I(U;Y^{(i)}|\bmY^{i-1}) \nonumber\\
&\leq \sum_{\bmy^{i-1}\in\calY^{i-1}} p_{\bmY^{i-1}}(\bmy^{i-1}) \sum_{i=1}^n p_U(u_i|\bmy^{i-1})  
\theta_{RR} D(p_{X^{(i)}|U=u_i, \bmY^{i-1}=\bmy^{i-1}} || p_{X^{(i)}|\bmY^{i-1}=\bmy^{i-1}}) \nonumber\\
& = \theta_{RR} I(U;X^{(i)}|\bmY^{i-1})
\label{eq:Comp_RR_IUYiYi-1_2}
\end{align}
By (\ref{eq:Comp_RR_IUY}) and (\ref{eq:Comp_RR_IUYiYi-1_2}), we obtain:
\begin{align}
I(U;Y) \leq \theta_{RR} I(U;X^{(1)}) + \sum_{i=2}^t \theta_{RR} I(U;X^{(i)}|\bmY^{i-1}).
\label{eq:Comp_RR_IUY_2}
\end{align}
In (\ref{eq:Comp_RR_IUY_2}), $I(U;X^{(i)}|\bmY^{i-1}) \leq H(U|\bmY^{i-1}) \leq \log n$ and $I(U;X^{(i)}| \allowbreak \bmY^{i-1}) \leq H(X^{(i)}|\bmY^{i-1}) \leq \log |\calX|$. 
Therefore, the combined release $Y$ provides 
$(n,t\alpha)$-PIE privacy, 
where $\alpha = \theta_{RR} \allowbreak \min\{\log n, \allowbreak \log |\calX|\}$.
\end{proof}

\begin{proof}[Proof of Theorem~\ref{thm:PIE_GLH} (PIE of the GLH)]
Let $p_U$ be a probability distribution of $U$. 
Let $p_Y$ be a probability distribution of $Y$, and $p_{Y|U=u_i}$ be a conditional probability distribution of $Y$ given $U=u_i$. 
Then $I(U;Y)$ can be written as follows:
\begin{align}
I(U;Y) = \sum_{i=1}^n p_U(u_i) D(p_{Y|U=u_i}||p_Y). 
\label{eq:PIE_GLH_IUY}
\end{align}
The $(g,\epsilon)$-GLH outputs $(h,y)\in\calY$ with probability in (\ref{eq:GLH}). 
We denote the probability that the $(g,\epsilon)$-GLH outputs $h$, $y$, and $(h,y)$ simply by $\Pr(h)$, $\Pr(y)$, and $\Pr(h,y)$, respectively. 
Then, $D(p_{Y|U=u_i}||p_Y)$ in (\ref{eq:PIE_GLH_IUY}) can be written as follows:
\begin{align}
&D(p_{Y|U=u_i}||p_Y) \nonumber\\
&= \sum_{(h,y)\in\calY} \Pr(h,y|U=u_i) \log \frac{\Pr(h,y|U=u_i)}{\Pr(h,y)} \nonumber\\
&= \hspace{-1mm} \sum_{h\in\calH} \Pr(h) \hspace{-1mm} \sum_{y\in[g]} \Pr(y|h,U=u_i) \log \frac{\Pr(h)\Pr(y|h,U=u_i)}{\Pr(h)\Pr(y|h)} \nonumber\\
& 
\hspace{4.8cm} 
\text{(by $\Pr(h|U=u_i) = \Pr(h)$)} \nonumber\\
&= \sum_{h\in\calH} \Pr(h) \sum_{y\in[g]} \Pr(y|h,U=u_i) \log \frac{\Pr(y|h,U=u_i)}{\Pr(y|h)}.
\label{eq:PIE_GLH_DpYuipY}
\end{align}
Let $q_{\text{uni}}$ be a uniform distribution over $[g]$.  
Given $x\in\calX$, let $z$ be a hash value $z = h(x) \in [g]$, and $Z$ be a random variable representing a hash value.  
Let $H$ be a random variable representing a hash function. 
Furthermore, let $p_{Z|H=h,U=u_i}$ (resp.~$p_{Z|H=h}$) be a conditional probability distribution of $Z$ given $H=h$ and $U=u_i$ (resp.~given $H=h$). 
By (\ref{eq:GLH}), the $(g,\epsilon)$-GLH outputs $y=z$ with probability $\theta_{GLH}$ ($= \frac{e^\epsilon-1}{g+e^\epsilon-1}$) and outputs a value randomly chosen from $[g]$ (including $z$) with probability $1-\theta_{GLH}$. 
Therefore, given $H=h$ and $U=u_i$, $Y$ is generated from $p_{Z|H=h,U=u_i}$ with probability $\theta_{GLH}$, and is generated from $q_{\text{uni}}$ with probability $1-\theta_{GLH}$. 
Then $D(p_{Y|U=u_i}||p_Y)$ in (\ref{eq:PIE_GLH_DpYuipY}) can be written as follows:
\begin{align}
&D(p_{Y|U=u_i}||p_Y) \nonumber\\
&= \sum_{h\in\calH} \Pr(h) D(\theta_{GLH}~ p_{Z|H=h,U=u_i} + (1-\theta_{GLH}) q_{\text{uni}} 
||\theta_{GLH}~ p_{Z|H=h} + (1-\theta_{GLH}) q_{\text{uni}}) \nonumber\\
&\leq \sum_{h\in\calH} \Pr(h) \theta_{GLH} D(p_{Z|H=h,U=u_i} || p_{Z|H=h}) \nonumber\\
& \hspace{5mm} \text{(by the convexity of the KL divergence \cite{elements} and 
$D(q_{\text{uni}}||q_{\text{uni}})=0$)} \nonumber\\
&= \theta_{GLH} \sum_{h\in\calH} \Pr(h) D(p_{Z|H=h,U=u_i} || p_{Z|H=h}).
\label{eq:PIE_GLH_DpYuipY_2}
\end{align}
Note that by expanding $I(U;H,Z)$ in the same way as (\ref{eq:PIE_GLH_IUY}) and (\ref{eq:PIE_GLH_DpYuipY}), $I(U;H,Z)$ can be expressed as follows:
\begin{align}
&I(U;H,Z) \nonumber\\
&= \sum_{i=1}^n p_U(u_i) \sum_{h\in\calH} \Pr(h) 
\sum_{z\in[g]} \Pr(z|h,U=u_i) \log \frac{\Pr(z|h,U=u_i)}{\Pr(z|h)} \nonumber\\
&= \sum_{i=1}^n p_U(u_i) \sum_{h\in\calH} \Pr(h) D(p_{Z|H=h,U=u_i} || p_{Z|H=h}).
\label{eq:PIE_GLH_IUX}
\end{align}
By (\ref{eq:PIE_GLH_IUY}), (\ref{eq:PIE_GLH_DpYuipY_2}), and (\ref{eq:PIE_GLH_IUX}), $I(U;Y)$ can be written as follows:
\begin{align}
&I(U;Y) \nonumber\\
&\leq \theta_{GLH} I(U;H,Z) \nonumber\\
&\leq \theta_{GLH} I(U;X) \hspace{2mm} \text{(by the data processing inequality \cite{elements}).}
\label{eq:PIE_GLH_IUY_2}
\end{align}
Therefore, (\ref{eq:GLH_IUY_IUX}) holds. 
This holds for any distribution 
$p_{U,X}$. 
(\ref{eq:PIE_GLH}) also holds because $I(U;X) \leq H(U) \leq \log n$ and $I(U;X) \leq H(X) \leq \log |\calX|$.
\end{proof}

\begin{proof}[Proof of Theorem~\ref{thm:Comp_GLH} (Composition of the GLH)]
For $i\in[t]$, let $\bmY^i = (Y^{(1)},\cdots,Y^{(i)})$. 
Let $H^{(i)}$ and $Z^{(i)}$ be random variables representing the $i$-th hash function and hash value, respectively. 

By the chain rule for the mutual information \cite{elements}, we obtain:
\begin{align}
I(U;Y) = I(U;Y^{(1)}) + \sum_{i=2}^t I(U;Y^{(i)}|\bmY^{i-1}).
\label{eq:Comp_GLH_IUY}
\end{align}
In the proof of Theorem~\ref{thm:Comp_RR}, we showed (\ref{eq:Comp_RR_IUYiYi-1}), (\ref{eq::Comp_RR_DpYi}), and (\ref{eq:Comp_RR_IUYiYi-1_2}). They differ from (\ref{eq:PIE_RR_IUY}), (\ref{eq:PIE_RR_D_pYuiY}), and  (\ref{eq:PIE_RR_IUY_2}), respectively, in that 
they 
add $\bmY^{i-1}$ as a condition. 
Similarly, by 
adding $\bmY^{i-1}$ as a condition in (\ref{eq:PIE_GLH_IUY}), (\ref{eq:PIE_GLH_DpYuipY}), (\ref{eq:PIE_GLH_DpYuipY_2}), and (\ref{eq:PIE_GLH_IUX}), we can show the following inequality:
\begin{align}
&I(U;Y^{(i)}|\bmY^{i-1}) \nonumber\\
&\leq \theta_{GLH} I(U;H^{(i)},Z^{(i)}|\bmY^{i-1}) \nonumber\\
&\leq \theta_{GLH} I(U;X^{(i)}|\bmY^{i-1}) \nonumber\\
&\hspace{4mm} \text{(by the data processing inequality \cite{elements}).}
\label{eq:Comp_GLH_IUYi}
\end{align}
By (\ref{eq:Comp_GLH_IUY}) and (\ref{eq:Comp_GLH_IUYi}), we obtain:
\begin{align}
I(U;Y) \leq \theta_{GLH} I(U;X^{(1)}) + \sum_{i=2}^t \theta_{GLH} I(U;X^{(i)}|\bmY^{i-1}).
\label{eq:Comp_GLH_IUY_2}
\end{align}
In (\ref{eq:Comp_GLH_IUY_2}), $I(U;X^{(i)}|\bmY^{i-1}) \leq H(U|\bmY^{i-1}) \leq \log n$ and $I(U;X^{(i)}| \allowbreak \bmY^{i-1}) \leq H(X^{(i)}|\bmY^{i-1}) \leq \log |\calX|$. 
Therefore, the combined release $Y$ provides 
$(n,t\alpha)$-PIE privacy, 
where $\alpha = \theta_{GLH} \allowbreak \min\{\log n, \allowbreak \log |\calX|\}$.
\end{proof}

\begin{proof}[Proof of Proposition~\ref{prop:l2_RR} ($l_2$ loss of the RR)]
Let 
\arxiv{$\hbmp_{RR}^*(x) =$ \allowbreak $\bbE[\hbmp_{RR}(x)]$}\conference{$\hbmp_{RR}^*(x) = \bbE[\hbmp_{RR}(x)]$}. 
Then the expected $l_2$ loss can be written as follows: 
\begin{align}
&\bbE[(\bmp(x) - \hbmp_{RR}(x))^2] \nonumber\\
&= \bbE \left[((\hbmp_{RR}(x) - \hbmp_{RR}^*(x)) + (\hbmp_{RR}^*(x) - \bmp(x)))^2 \right] \nonumber\\
&= \bbE \left[(\hbmp_{RR}(x) - \hbmp_{RR}^*(x))^2 \right] + (\hbmp_{RR}^*(x) - \bmp(x)))^2 
\nonumber\\
& \hspace{3mm} 
+ 2(\hbmp_{RR}^*(x) - \bmp(x)) \bbE \left[(\hbmp_{RR}(x) - \hbmp_{RR}^*(x)) \right] \nonumber\\
&= \bbE \left[(\hbmp_{RR}(x) - \hbmp_{RR}^*(x))^2 \right] + (\hbmp_{RR}^*(x) - \bmp(x)))^2.
\label{eq:l2_RR_bias_variance}
\end{align}
The first term in (\ref{eq:l2_RR_bias_variance}) is the variance, whereas the second term in (\ref{eq:l2_RR_bias_variance}) is the bias. 
In other words, the expected $l_2$ loss consists of the bias and variance, which is called the bias-variance trade-off \cite{mlpp}.
It is well known that the empirical estimate is unbiased; i.e., $\hbmp_{RR}^*(x) = \bmp(x)$ \cite{Kairouz_ICML16,Wang_USENIX17}. 
Therefore, we obtain:
\begin{align}
\bbE[(\bmp(x) - \hbmp_{RR}(x))^2] 
&= \bbE \left[(\hbmp_{RR}(x) - \hbmp_{RR}^*(x))^2 \right] \nonumber\\
&= \bbE \left[(\hbmp_{RR}(x) - \bmp(x))^2 \right].
\label{eq:l2_RR_l2_variance}
\end{align}
By (\ref{eq:RR_emp}), $\mu_{RR} = \frac{e^\epsilon}{|\calX|+e^\epsilon-1}$, and $\nu_{RR} = \frac{1}{|\calX|+e^\epsilon-1}$, the right side of (\ref{eq:l2_RR_l2_variance}) can be written as follows: 
\begin{align}
\bbE \left[(\hbmp_{RR}(x) - \bmp(x))^2 \right] 
&=\left( \frac{1}{\mu_{RR} - \nu_{RR}} \right)^2  \frac{\text{Var}[\bmc_{RR}(x)]}{n^2} \nonumber\\
&=\left( \frac{|\calX|+e^\epsilon-1}{e^\epsilon-1} \right)^2 \frac{\text{Var}[\bmc_{RR}(x)]}{n^2},
\label{eq:l2_RR_E_hbmp_bmp}
\end{align}
where for $a\in\reals$, $\text{Var}[a]$ represents the variance of $a$. 

Recall that $\bmc_{RR}(x)$ is the number of $x\in\calX$ in $\bmY_{1:n}=(Y_1,\cdots,Y_n)$. 
For $i\in[n]$, let $b_i(x)\in\{0,1\}$ be a value that takes $1$ if $Y_i=x$ and $0$ if $Y_i\ne x$. 
Then $\bmc_{RR}(x) = \sum_{i=1}^n b_i(x)$. 
By (\ref{eq:RR}), $b_i(x)$ is randomly generated from the Bernoulli distribution with parameter 
$\mu_{RR}$ (resp.~$\nu_{RR}$) 
if $X_i=x$ (resp.~$X_i\ne x$). 
The number of $x$ in $\bmX_{1:n}$ is $n\bmp(x)$, whereas the number of the other input 
symbols 
in $\bmX_{1:n}$ is $n(1-\bmp(x))$. 
Therefore, $\text{Var}[\bmc_{RR}(x)]$ in (\ref{eq:l2_RR_E_hbmp_bmp}) can be written as follows:
\begin{align}
&\text{Var}[\bmc_{RR}(x)] \nonumber\\
&= n\bmp(x)\mu_{RR}(1-\mu_{RR}) + n(1-\bmp(x))\nu_{RR}(1-\nu_{RR}) \nonumber\\
&= \frac{n\bmp(x)e^\epsilon(|\calX|-1)}{(|\calX|+e^\epsilon-1)^2} + \frac{n(1-\bmp(x))(|\calX|+e^\epsilon-2)}{(|\calX|+e^\epsilon-1)^2}.
\label{eq:l2_RR_Var_c_RR}
\end{align}
By (\ref{eq:l2_RR_l2_variance}), (\ref{eq:l2_RR_E_hbmp_bmp}), and (\ref{eq:l2_RR_Var_c_RR}), we obtain:
\begin{align*}
&\bbE[(\bmp(x) - \hbmp_{RR}(x))^2] \\
&= \frac{\bmp(x)e^\epsilon(|\calX|-1) + (1-\bmp(x))(|\calX|+e^\epsilon-2)}{n(e^\epsilon-1)^2} \\
&= \frac{|\calX|+e^\epsilon-2}{n(e^\epsilon-1)^2} + \frac{\bmp(x)e^\epsilon(|\calX|-2) - \bmp(x)(|\calX|-2)}{n(e^\epsilon-1)^2} \\
&= \frac{|\calX|+e^\epsilon-2}{n(e^\epsilon-1)^2} + \frac{\bmp(x)(|\calX|-2)}{n(e^\epsilon-1)}. 
\end{align*}
Therefore (\ref{eq:l2_RR}) holds.
\end{proof}

\begin{proof}[Proof of Proposition~\ref{prop:l2_GLH} ($l_2$ loss of the GLH)]
As with (\ref{eq:l2_RR_l2_variance}), the expected $l_2$ loss of the GLH is equal to the variance of the GLH: 
\begin{align}
\bbE[(\bmp(x) - \hbmp_{GLH}(x))^2] 
&= \bbE \left[(\hbmp_{GLH}(x) - \bmp(x))^2 \right].
\label{eq:l2_GLH_l2_variance}
\end{align}
By (\ref{eq:GLH_emp}), $\mu_{GLH}=\frac{e^\epsilon}{g+e^\epsilon-1}$, and $\nu_{GLH}=\frac{1}{g}$, the right side of (\ref{eq:l2_GLH_l2_variance}) can be written as follows: 
\begin{align}
&\bbE \left[(\hbmp_{GLH}(x) - \bmp(x))^2 \right] \nonumber\\
&=\left( \frac{1}{\mu_{GLH} - \nu_{GLH}} \right)^2  \frac{\text{Var}[\bmc_{GLH}(x)]}{n^2} \nonumber\\
&= \left( \frac{g(g+e^\epsilon-1)}{ge^\epsilon-(g+e^\epsilon-1)} \right)^2 \frac{\text{Var}[\bmc_{GLH}(x)]}{n^2} \nonumber\\
&= \left( \frac{g(g+e^\epsilon-1)}{(g-1)(e^\epsilon-1)} \right)^2 \frac{\text{Var}[\bmc_{GLH}(x)]}{n^2}.
\label{eq:l2_GLH_E_hbmp_bmp}
\end{align}
Recall that $\bmc_{GLH}(x)$ is the number of $(h,y) \in \calY_x$ in $\bmY_{1:n}=(Y_1,\cdots,\allowbreak Y_n)$ and 
$\calY_x = \{(h,y)\in\calY | y=h(x)\}$. 
For $i\in[n]$, let $b'_i(x)\in\{0,1\}$ be a value that takes $1$ if $Y_i$ is in $\calY_x$ and $0$ otherwise. 
Then $\bmc_{GLH}(x) = \sum_{i=1}^n b'_i(x)$. 
Assume that $X_i \neq x$. 
In this case, $b_i(x)$ is $1$ if the output of the hash function $h$ with input $X_i$ collides with $h(x)$, which happens with probability $\nu_{GLH} = \frac{1}{g}$. 
Then by (\ref{eq:GLH}), $b'_i(x)$ is randomly generated from the Bernoulli distribution with parameter $\mu_{GLH}$ (resp.~$\nu_{GLH}$) if $X_i=x$ (resp.~$X_i\ne x$). 
The number of $x$ in $\bmX_{1:n}$ is $n\bmp(x)$, whereas the number of the other input 
symbols 
in $\bmX_{1:n}$ is $n(1-\bmp(x))$. 
Therefore, $\text{Var}[\bmc_{GLH}(x)]$ in (\ref{eq:l2_GLH_E_hbmp_bmp}) can be written as follows:
\begin{align}
&\text{Var}[\bmc_{GLH}(x)] \nonumber\\
&= n\bmp(x)\mu_{GLH}(1-\mu_{GLH}) + n(1-\bmp(x))\nu_{GLH}(1-\nu_{GLH}) \nonumber\\
&= \frac{n\bmp(x)e^\epsilon(g-1)}{(g+e^\epsilon-1)^2} + \frac{n(1-\bmp(x))(g-1)}{g^2}.
\label{eq:l2_GLH_Var_c_GLH}
\end{align}
By (\ref{eq:l2_GLH_l2_variance}), (\ref{eq:l2_GLH_E_hbmp_bmp}), and (\ref{eq:l2_GLH_Var_c_GLH}), we obtain:
\begin{align*}
&\bbE[(\bmp(x) - \hbmp_{GLH}(x))^2] \\
&= \frac{\bmp(x)g^2 e^\epsilon}{n(e^\epsilon-1)^2(g-1)} + \frac{(1-\bmp(x))(g+e^\epsilon-1)^2}{n(e^\epsilon-1)^2(g-1)} \\
&= \frac{(g+e^\epsilon-1)^2}{n(e^\epsilon-1)^2(g-1)} 
+ \frac{\bmp(x)(g^2e^\epsilon - g^2 - 2g(e^\epsilon-1) - (e^\epsilon-1)^2)}{n(e^\epsilon-1)^2(g-1)} \\
&= \frac{(g+e^\epsilon-1)^2}{n(e^\epsilon-1)^2(g-1)} + \frac{\bmp(x)(g^2 - 2g - e^\epsilon + 1)}{n(e^\epsilon-1)(g-1)}.
\end{align*}
Therefore (\ref{eq:l2_GLH}) holds.
\end{proof}

\begin{proof}[Proof of Theorem~\ref{thm:optimal_g} (Optimal $g$ in the GLH)]
By (\ref{eq:l2_GLH}) and \allowbreak $\theta_{GLH} = \frac{e^\epsilon-1}{g+e^\epsilon-1}$, the expected $l_2$ loss can be written as follows:
\begin{align}
&\bbE[(\bmp(x) - \hbmp_{GLH}(x))^2] \nonumber\\
&= \frac{(g+e^\epsilon-1)^2}{n(e^\epsilon-1)^2(g-1)} + \frac{\bmp(x)(g^2 - 2g - e^\epsilon + 1)}{n(e^\epsilon-1)(g-1)} \nonumber\\
&= \frac{(g+e^\epsilon-1)^2}{n(e^\epsilon-1)^2(g-1)} + \frac{\bmp(x)g(g-1) }{n(e^\epsilon-1)(g-1)} - \frac{\bmp(x)(g + e^\epsilon - 1)}{n(e^\epsilon-1)(g-1)} \nonumber\\
&= \frac{1}{n\theta_{GLH}^2(g-1)} + \frac{\bmp(x)g}{n(e^\epsilon-1)} - \frac{\bmp(x)}{n\theta_{GLH}(g-1)} \nonumber\\
&= \frac{1-\bmp(x)\theta_{GLH}}{n\theta_{GLH}^2(g-1)} + \frac{\bmp(x)g}{n(e^\epsilon-1)} 
\label{eq:optimal_g_l2_loss}
\end{align}
By $\theta_{GLH} = \frac{e^\epsilon-1}{g+e^\epsilon-1}$ and $1-\theta_{GLH} = \frac{g}{g+e^\epsilon-1}$, (\ref{eq:optimal_g_l2_loss}) can be expressed as follows: 
\begin{align}
\bbE[(\bmp(x) - \hbmp_{GLH}(x))^2] = \frac{1-\bmp(x)\theta_{GLH}}{n\theta_{GLH}^2(g-1)} + \frac{\bmp(x)(1-\theta_{GLH})}{n\theta_{GLH}}.
\label{eq:optimal_g_l2_loss_2}
\end{align}
Now we consider increasing $g$ while fixing $\theta_{GLH}$ in (\ref{eq:optimal_g_l2_loss_2}). 
Note that $0 \leq \bmp(x) \leq 1$ and $0 \leq \theta_{GLH} \leq 1$. 
Moreover, $n$ and $\bmp(x)$ do not depend on $g$. 
Therefore, (\ref{eq:optimal_g_l2_loss_2}) is monotonically decreasing in $g$. 
In addition, the first term in (\ref{eq:optimal_g_l2_loss_2}) goes to zero as $g$ goes to infinity. 
Therefore, 
\begin{align*}
\bbE[(\bmp(x) - \hbmp_{GLH}(x))^2] \rightarrow \frac{\bmp(x)(1-\theta_{GLH})}{n\theta_{GLH}} \hspace{3mm} (g \rightarrow \infty),
\end{align*}
which proves Theorem~\ref{thm:optimal_g}.
\end{proof}

\end{document}